\documentclass[11pt,notitlepage,tightenlines,nofootinbib,superscriptaddress,aps,pra]{revtex4-2}
\pdfoutput=1

\usepackage{newpxtext,newpxmath}

\usepackage[latin1]{inputenc}
\usepackage{amsthm}
\usepackage{amssymb}
\usepackage{amsmath}
\usepackage{bbold}
\usepackage{bbm}
\usepackage{braket}
\usepackage{dsfont}
\usepackage{mathdots}
\usepackage{mathtools}
\usepackage{enumerate}
\usepackage[shortlabels]{enumitem}
\usepackage{csquotes}
\usepackage{stmaryrd}
\usepackage[cal=boondox]{mathalfa}
\usepackage{graphicx}
\usepackage{stackengine}
\usepackage{scalerel}
\usepackage{tensor}       
\usepackage[dvipsnames,svgnames]{xcolor}
\usepackage{array}
\usepackage{makecell}
\newcolumntype{x}[1]{>{\centering\arraybackslash}p{#1}}
\usepackage{tikz}
\usepackage{pgfplots}
\usetikzlibrary{shapes.geometric, shapes.misc, positioning, arrows, arrows.meta, decorations.pathreplacing, decorations.pathmorphing, patterns, angles, quotes, calc}
\usepackage{booktabs}
\usepackage{xfrac}
\usepackage{siunitx}
\usepackage{centernot}
\usepackage{comment}
\usepackage{chngcntr}

\usepackage[pdftex]{hyperref}
\hypersetup{
    bookmarksnumbered=true, 
    breaklinks=true,
    unicode=false, 
    pdfstartview={FitH}, 
    pdfnewwindow=true, 
    colorlinks=true, 
    allcolors=MidnightBlue!70!black!70!TealBlue
}

\newtheorem{thm}{Theorem}
\newtheorem*{thm*}{Theorem}
\newtheorem{prop}[thm]{Proposition}
\newtheorem*{prop*}{Proposition}
\newtheorem{lemma}[thm]{Lemma}
\newtheorem*{lemma*}{Lemma}
\newtheorem{cor}[thm]{Corollary}
\newtheorem*{cor*}{Corollary}

\newtheorem*{cj*}{Conjecture}
\newtheorem{Def}[thm]{Definition}
\newtheorem*{Def*}{Definition}

\newtheorem*{question*}{Question}

\newtheorem*{problem*}{Problem}

\makeatletter
\def\thmhead@plain#1#2#3{%
  \thmname{#1}\thmnumber{\@ifnotempty{#1}{ }\@upn{#2}}%
  \thmnote{ {\the\thm@notefont#3}}}
\let\thmhead\thmhead@plain
\makeatother

\theoremstyle{definition}
\newtheorem{rem}[thm]{Remark}
\newtheorem*{rem*}{Remark}

\newcommand{\bb}{\begin{equation}\begin{aligned}\hspace{0pt}}
\newcommand{\bbb}{\begin{equation*}\begin{aligned}}
\newcommand{\ee}{\end{aligned}\end{equation}}
\newcommand{\eee}{\end{aligned}\end{equation*}}
\newcommand\floor[1]{\lfloor#1\rfloor}
\newcommand\ceil[1]{\left\lceil#1\right\rceil}
\newcommand{\eqt}[1]{\stackrel{\mathclap{\mbox{\scriptsize#1}}}{=}}
\newcommand{\leqt}[1]{\stackrel{\mathclap{\mbox{\scriptsize#1}}}{\leq}}

\newcommand{\geqt}[1]{\stackrel{\mathclap{\mbox{\scriptsize#1}}}{\geq}}
\newcommand{\ketbra}[1]{\ket{#1}\!\!\bra{#1}}

\newcommand{\sumno}{\sum\nolimits}

\newcommand{\e}{\varepsilon}
\renewcommand{\epsilon}{\varepsilon}

\newcommand{\id}{\mathds{1}}
\newcommand{\R}{\mathds{R}}
\newcommand{\N}{\mathds{N}}

\newcommand{\C}{\mathds{C}}

\newcommand{\ve}{\varepsilon}

\newcommand{\SEP}{\pazocal{S}}
\newcommand{\PPT}{\pazocal{P\!P\!T}}

\DeclareMathOperator{\Tr}{Tr}

\DeclareMathOperator{\co}{conv}
\DeclareMathOperator{\cone}{cone}

\DeclareMathOperator{\Span}{span}
\DeclareMathAlphabet{\pazocal}{OMS}{zplm}{m}{n}

\DeclareMathOperator{\spec}{spec}

\DeclareMathOperator{\tr}{tr}

\newcommand{\HH}{\pazocal{H}}

\newcommand{\MM}{\pazocal{M}}
\newcommand{\D}{\pazocal{D}}

\newcommand{\XX}{\pazocal{X}}

\newcommand{\PP}{\pazocal{P}}
\newcommand{\FF}{\pazocal{F}}

\newcommand{\lsmatrix}{\left(\begin{smallmatrix}}
\newcommand{\rsmatrix}{\end{smallmatrix}\right)}

\newcommand{\deff}[1]{\textbf{\emph{#1}}}

\stackMath
\newcommand\xxrightarrow[2][]{\mathrel{%
  \setbox2=\hbox{\stackon{\scriptstyle#1}{\scriptstyle#2}}%
  \stackunder[5pt]{%
    \xrightarrow{\makebox[\dimexpr\wd2\relax]{$\scriptstyle#2$}}%
  }{%
   \scriptstyle#1\,%
  }%
}}

\newcommand{\tendsn}[1]{\xxrightarrow[\! n\rightarrow \infty\!]{#1}}

\newcommand{\ctends}[3]{\xxrightarrow[\raisebox{#3}{$\scriptstyle #2$}]{\raisebox{-0.7pt}{$\scriptstyle #1$}}}

\stackMath

\makeatletter
\newcommand*\rel@kern[1]{\kern#1\dimexpr\macc@kerna}
\newcommand*\widebar[1]{%
  \begingroup
  \def\mathaccent##1##2{%
    \rel@kern{0.8}%
    \overline{\rel@kern{-0.8}\macc@nucleus\rel@kern{0.2}}%
    \rel@kern{-0.2}%
  }%
  \macc@depth\@ne
  \let\math@bgroup\@empty \let\math@egroup\macc@set@skewchar
  \mathsurround\z@ \frozen@everymath{\mathgroup\macc@group\relax}%
  \macc@set@skewchar\relax
  \let\mathaccentV\macc@nested@a
  \macc@nested@a\relax111{#1}%
  \endgroup
}

\counterwithin*{equation}{part}
\counterwithin*{thm}{part}
\counterwithin*{figure}{part}

\tikzset{meter/.append style={draw, inner sep=10, rectangle, font=\vphantom{A}, minimum width=30, line width=.8, path picture={\draw[black] ([shift={(.1,.3)}]path picture bounding box.south west) to[bend left=50] ([shift={(-.1,.3)}]path picture bounding box.south east);\draw[black,-latex] ([shift={(0,.1)}]path picture bounding box.south) -- ([shift={(.3,-.1)}]path picture bounding box.north);}}}
\tikzset{roundnode/.append style={circle, draw=black, fill=gray!20, thick, minimum size=10mm}}
\tikzset{squarenode/.style={rectangle, draw=black, fill=none, thick, minimum size=10mm}}

\definecolor{Blues5seq1}{RGB}{239,243,255}
\definecolor{Blues5seq2}{RGB}{189,215,231}
\definecolor{Blues5seq3}{RGB}{107,174,214}
\definecolor{Blues5seq4}{RGB}{49,130,189}
\definecolor{Blues5seq5}{RGB}{8,81,156}

\definecolor{Greens5seq1}{RGB}{237,248,233}
\definecolor{Greens5seq2}{RGB}{186,228,179}
\definecolor{Greens5seq3}{RGB}{116,196,118}
\definecolor{Greens5seq4}{RGB}{49,163,84}
\definecolor{Greens5seq5}{RGB}{0,109,44}

\definecolor{Reds5seq1}{RGB}{254,229,217}
\definecolor{Reds5seq2}{RGB}{252,174,145}
\definecolor{Reds5seq3}{RGB}{251,106,74}
\definecolor{Reds5seq4}{RGB}{222,45,38}
\definecolor{Reds5seq5}{RGB}{165,15,21}

\allowdisplaybreaks

\let\nc\newcommand
\nc{\proj}[1]{\ket{#1}\!\bra{#1}}
\renewcommand{\bar}{\;\rule{0pt}{9.5pt}\right|\;}
\nc{\lset}{\left\{\left.}
\nc{\rset}{\right\}}
\nc{\lsetr}{\left\{\,}
\nc{\rsetr}{\right.\right\}}
\nc{\barr}{\;\rule{0pt}{9.5pt}\left|\;}

\nc{\Sanov}{\mathrm{Sanov}}
\nc{\Stein}{\mathrm{Stein}}

\makeatletter
\renewenvironment{proof}[1][\proofname]{\par
\pushQED{\qed}%
\normalfont \topsep6\p@\@plus6\p@\relax
\trivlist
\item\relax
{\bfseries  
#1\@addpunct{.}}\hspace\labelsep\ignorespaces 
}{%
\popQED\endtrivlist\@endpefalse
}
\makeatother

\makeatletter
\renewcommand{\p@subsection}{\thesection.}
\renewcommand{\p@subsubsection}{\thesection.\thesubsection.}
\makeatother

\newcommand{\rel}[3]{#1\big(#2\,\big\|\,#3\big)}

\newcommand{\LL}{\pazocal{L}}
\renewcommand{\AA}{\pazocal{A}}
\newcommand{\BB}{\pazocal{B}}
\newcommand{\all}{\mathds{ALL}}

\newcommand{\sanov}{\mathrm{Sanov}}
\begin{document}


\title{Asymptotic quantification of entanglement with a single copy}

\author{Ludovico Lami}
\email{ludovico.lami@gmail.com}
\affiliation{Scuola Normale Superiore, Piazza dei Cavalieri 7, 56126 Pisa, Italy}
\affiliation{QuSoft, Science Park 123, 1098 XG Amsterdam, the Netherlands}
\affiliation{Korteweg--de Vries Institute for Mathematics, University of Amsterdam, Science Park 105-107, 1098 XG Amsterdam, the Netherlands}
\affiliation{Institute for Theoretical Physics, University of Amsterdam, Science Park 904, 1098 XH Amsterdam, the Netherlands}

\author{Mario Berta}
\email{berta@physik.rwth-aachen.de}
\affiliation{Institute for Quantum Information, RWTH Aachen University, Aachen, Germany}

\author{Bartosz Regula}
\email{bartosz.regula@gmail.com}
\affiliation{Mathematical Quantum Information RIKEN Hakubi Research Team, RIKEN Pioneering Research Institute (PRI) and RIKEN Center for Quantum Computing (RQC), Wako, Saitama 351-0198, Japan}


\begin{abstract}
Despite the central importance of quantum entanglement in quantum technologies, the understanding of the optimal ways to exploit it is still beyond our reach, and even measuring entanglement in an operationally meaningful way is prohibitively difficult. Here we study two fundamental tasks in the processing of entanglement: entanglement testing, which is a quantum state discrimination problem concerned with entanglement detection in the many-copy regime, and entanglement distillation,
concerned with purifying entanglement from noisy entangled states. 
We introduce a way of benchmarking the performance of distillation that focuses on the best achievable error rather than its yield in the asymptotic limit. When the underlying set of operations used for entanglement distillation is the axiomatic class of non-entangling operations, we show that the two figures of merit for entanglement testing and distillation coincide. We solve both problems by proving a generalised quantum Sanov's theorem, enabling the exact evaluation of asymptotic error rates of composite quantum hypothesis testing. We show in particular that the asymptotic figure of merit is given by the reverse relative entropy of entanglement, a single-letter quantity that can be evaluated using only a single copy of a quantum state --- a distinct feature among measures of entanglement that quantify the optimal performance of information-theoretic tasks.
\end{abstract}

\maketitle


\let\oldaddcontentsline\addcontentsline
\renewcommand{\addcontentsline}[3]{}

\section{Introduction}

The phenomenon of quantum entanglement is one of the most important resources that underlie the potential of quantum technologies to provide advantages in information processing and computation~\cite{teleportation,dense-coding,Ekert91, RennerPhD,Brunner-review,Horodecki-review}. The understanding of how to process and use entanglement is crucial to its applications, but 
our knowledge of the optimal performance of operational tasks involving entanglement is still incomplete.
Two key examples of such problems, which turn out to be profoundly connected, are entanglement testing and entanglement distillation.

Entanglement testing can be understood as a type of entanglement detection. In this task, one wishes to certify whether an untrusted source that is supposed to generate copies of some entangled state $\rho_{AB}$ is performing as intended or, alternatively, is faulty and produces only separable (unentangled) states. 
Natural figures of merit for this task are based on the minimal probabilities of a misdetection, which could be either a false positive --- that is, mistaking a working device for a faulty one --- or, vice versa, a false negative. 
In light of this operational interpretation, any such metric can be interpreted as a measure of the entanglement content of $\rho_{AB}$: the more entangled the state is, the easier it is to distinguish from separable ones. 
Such questions are naturally characterised through the framework of composite quantum hypothesis testing~\cite{Brandao2010}, but despite active progress in the study of related problems~\cite{gap,hayashi_sanov,hayashi_stein,blurring}, obtaining a computable expression for the optimal performance in this task has been elusive.

Another key operational primitive of quantum information processing is
entanglement distillation, a task 
introduced in the pioneering works~\cite{Bennett-distillation, Bennett-distillation-mixed, Bennett-error-correction}, which aims to purify noisy entangled states into maximally entangled ones. This process is an important
ingredient in many practical quantum information protocols, as high-fidelity entanglement is typically a prerequisite for quantum computation and communication schemes. Moreover, entanglement distillation is deeply connected to the theory of quantum error correction~\cite{Bennett-error-correction}. Despite this clear significance, and despite it being one of the very first operational protocols ever studied in quantum information theory, we still do not have a complete understanding of entanglement distillation. Most notably, we lack a computable formula for how much entanglement can be distilled from a given quantum state, and even deciding whether any entanglement whatsoever can be extracted is an unsolved problem in general~\cite{Horodecki-open-problems, list-open-problems}. Similar problems affect other entanglement processing tasks, and exact solutions generally exist only in few special cases.

The main difficulty in the study of both of these tasks lies in the fact that 
their performance can typically be improved by employing more copies of a given quantum state, which means that the ultimate efficiency of a protocol needs to be understood in an asymptotic sense: given more and more copies of a given quantum state, how does the performance improve? This leads to a natural information-theoretic description of such tasks in terms of asymptotic rates, whose evaluation is the main bottleneck in the understanding of the operational properties of quantum entanglement.

An unfortunate consequence of this asymptotic character of entanglement processing is that, even when one can identify a relevant closed-form quantity that describes the given task --- such as, for example, the quantum relative entropy~\cite{Vedral1997} or the entanglement of formation~\cite{Bennett-error-correction, Wootters1998} --- the optimal asymptotic rate can only be expressed using so-called \emph{regularised} formulas, which require the evaluation of an explicit limit in the number of copies of the given quantum state $\rho_{AB}$~\cite{Werner-symmetry, Shor2004, Hastings2008, Hayden-EC, devetak2005}. This leads to expressions of the form $\lim_{n\to\infty} \frac1n f(\rho^{\otimes n}_{AB})$, which are immensely difficult to evaluate even for simple functions $f$, thus preventing an efficient quantitative characterisation of
the asymptotic operational properties of entanglement.
Because of this, the optimal rates of not only entanglement testing or 
distillation but also other important operational tasks remain inaccessible in general. When it comes to distillation, this 
problem persists not merely in the standard, practically motivated setting for 
manipulating
 entanglement --- namely, in the paradigm of local operations and classical communication (LOCC)~\cite{Bennett-distillation, Bennett-distillation-mixed, Bennett-error-correction} --- but even in simplified mathematical frameworks where entanglement manipulation is studied under relaxed constraints~\cite{Rains2001, Horodecki2002, BrandaoPlenio1}, which are designed to provide a more tractable structure for studying entanglement conversion.

One may then wonder: if precise answers are so hard to find, could we instead obtain insights into the asymptotic properties of quantum entanglement
by adjusting the questions that we ask? More specifically, while traditional approaches to entanglement processing remain fundamental and key for many applications, could we obtain a simpler asymptotic characterisation of 
these tasks by changing the way in which we benchmark the performance of protocols, shifting the focus to another figure of merit? This question will motivate the core of our approach.

In the setting of entanglement testing, this will entail a seemingly minor change of focus from the asymptotic probability of a false negative (type~II error exponent), which is what most prior works were concerned with~\cite{Brandao2010,gap,hayashi_stein,blurring}, to the asymptotic probability a false positive (type~I error exponent), for which a closed-form solution was unknown prior to our work~\cite{hayashi_sanov}. 
The study of these two deceptively similar variants of the problem requires 
conceptually different techniques, and --- crucially --- we will see that this modification will lead to a remarkable simplification of the resulting expression.

To characterise entanglement distillation, we propose a conceptual shift: instead of focusing on the quantity (yield) as the measure of the efficiency of the protocol when more and more copies of a given state are available, we will focus on the \emph{quality} of the obtained entanglement, which is represented by the optimal error exponent --- that is, the rate at which the error of the protocol can be decreased.
This approach is inspired by the information-theoretic characterisation of quantum hypothesis testing~\cite{Hiai1991, Ogawa2000, HAYASHI}, where it is precisely this exponent of error probability that constitutes the figure of merit. While any entanglement manipulation framework can be studied through this lens, here we will focus on the one defined by the axiomatic class of non-entangling operations~\cite{BrandaoPlenio1} --- this useful relaxation of the operationally motivated LOCC framework has been used to shed light on the connections between entanglement theory and thermodynamics~\cite{BrandaoPlenio1, irreversibility, probabilistic-reversibility, hayashi_stein, blurring}, and its axiomatic nature means that it can be generalised even to quantum resources beyond entanglement~\cite{Brandao-Gour,probabilistic-reversibility, hayashi_stein, blurring}. The simpler structure of these operations will allow us to obtain an exact asymptotic solution.

Our first result establishes an exact equivalence between 
the performance of the two tasks discussed above, namely entanglement testing (in its standard formulation, under all physically realisable measurements) and entanglement distillation (under non-entangling operations): we show that the error exponent of entanglement distillation equals the exponent of false positive error in entanglement testing. 
This connection will allow us to tackle both of these tasks at the same time through an information-theoretic study of the underlying hypothesis testing problem.
Indeed, computing the asymptotic exponent of entanglement testing is a generalisation of a result in quantum hypothesis testing known as quantum Sanov's theorem~\cite{bjelakovic_sanov,notzel_sanov,hayashi_sanov}. However, the much more complicated structure involved in the problem we encounter here means that no known results are sufficiently general to shed any light on it. The problem is also related to the generalised quantum Stein's lemma~\cite{Brandao2010,hayashi_stein,blurring} that attracted much attention recently, but its distinct structure means that it requires a different approach.

As our main contribution, we then establish a generalised quantum Sanov's theorem that yields an exact expression for the asymptotic performance of entanglement testing and, as a result, also for the error exponent of entanglement distillation under non-entangling operations. We in particular show that the exponent is given by a variant of the relative entropy of entanglement, the \emph{reverse} relative entropy of entanglement~\cite{Vedral1997, eisert_reverse}. The remarkable aspect of this result is that the quantity can be evaluated exactly --- without regularisation --- on a single copy of the given quantum state, circumventing the problems that affect other measures of entanglement connected with practical tasks.
Our result thus establishes the reverse relative entropy as a measure of entanglement with a twofold direct meaning, 
while at the same time being computable without having to evaluate a many-copy limit. This altogether gives an exact solution to the problem of entanglement testing and provides 
an alternative way of benchmarking entanglement distillation, 
avoiding the seemingly ubiquitous problem of regularised formulas in the quantification of the performance of asymptotic entanglement processing protocols and thus bypassing the resulting bottlenecks.



\section{Entanglement testing and distillation}


\subsection{Setting of entanglement testing}

In the basic scenario of quantum hypothesis testing (quantum state discrimination),  one is tasked with distinguishing between two quantum states $\rho$ and $\sigma$ by performing a collective measurement on $n$ copies of the unknown state. The probability of mistaking $\rho$ for $\sigma$ decays exponentially as $2^{-cn}$, and it is this \emph{error exponent} $c$ that one aims to quantify in order to understand how fast the distinguishability improves as more copies become available. Remarkably, in the limit as $n\to\infty$, the error exponent exactly equals the quantum relative entropy $D(\sigma \| \rho) = \Tr \!\left[\sigma (\log_2 \sigma - \log_2 \rho)\right]$~\cite{Hiai1991,Ogawa2000}. It is this result, known as quantum Stein's lemma, that gives the quantum relative entropy its operational meaning as a measure of distinguishability of quantum states.

Consider now a scenario where two separated parties, Alice and Bob, would like to use a device that is supposed to prepare $n$ copies of some entangled state $\rho_{AB}$. 
However, they suspect that the device may fail, preparing a state that has no entanglement whatsoever between Alice's and Bob's systems.
How to verify if one really obtained the desired entangled state? This task, which we call entanglement testing, can be phrased as a composite hypothesis testing problem~\cite{Brandao2010}: we are to distinguish between $\rho_{AB}^{\otimes n}$ and the whole set of separable quantum states with a measurement (see Figure~\ref{fig:entanglement_testing}). Just as in conventional hypothesis testing, we would like to understand the behaviour of the optimal error exponent for large $n$, where the optimisation refers to the discrimination strategies. 
In order for this to characterise the optimal performance of the most general discrimination schemes, 
we do not impose any a priori constraints on the kind of measurement that can be carried out on the system, meaning that the above optimisation is understood to run over all physically realisable quantum measurements. This, in turn, makes the error exponent difficult to control and constitutes the main challenge in understanding asymptotic entanglement testing.

\begin{figure}[h]
\includegraphics[width=0.9\textwidth]{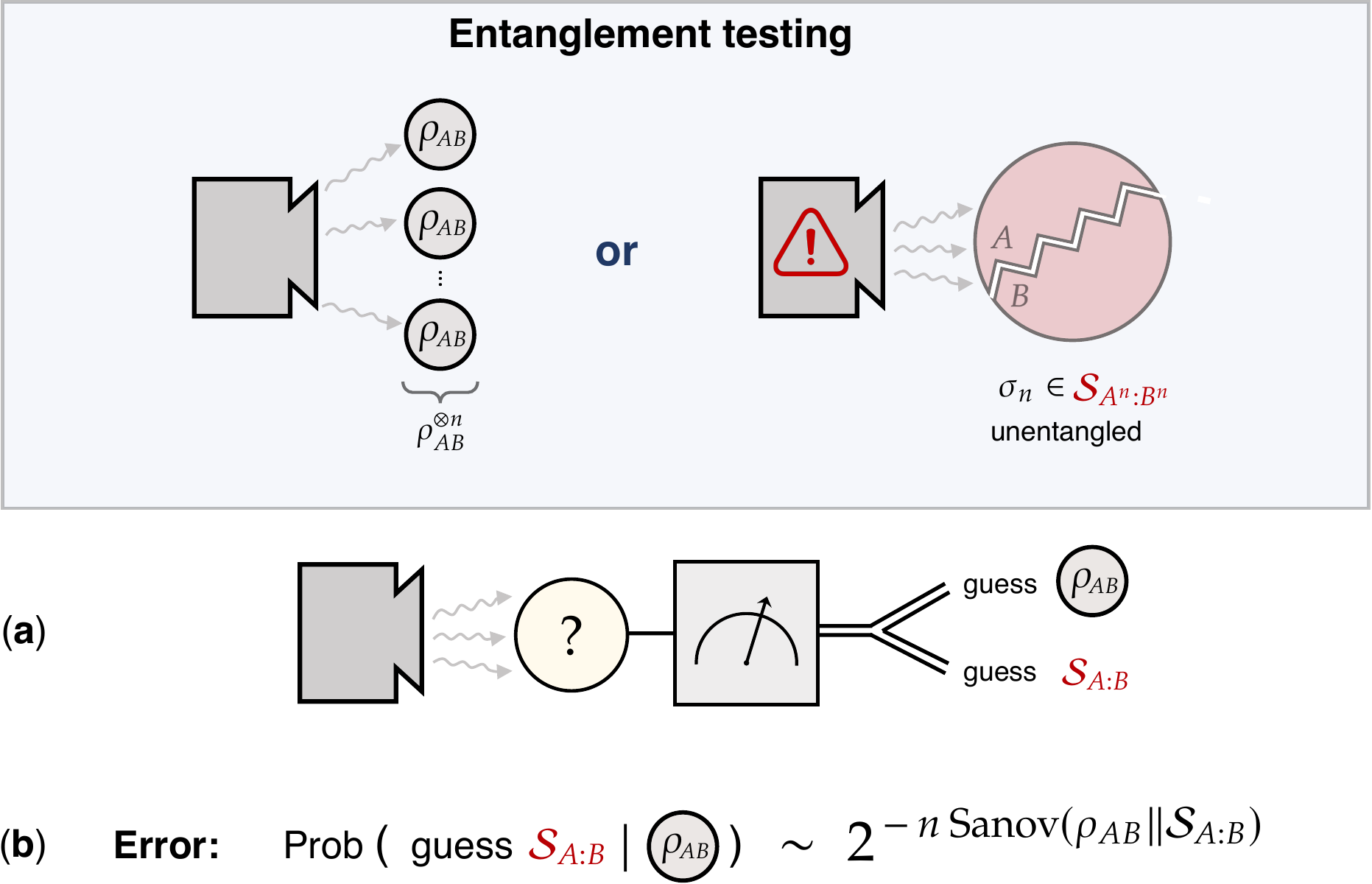}
\caption{%
\textbf{The set-up and figure of merit in entanglement testing.} Entanglement testing is a quantum hypothesis testing problem concerned with distinguishing the case when a source is generating copies of a target entangled state $\rho_{AB}$ from the case when it malfunctions and instead produces only states $\sigma_n\in \SEP_{A^n:B^n}$ that are globally separable, i.e.\ exhibit no entanglement between between Alice's systems on one side and Bob's systems on the other. (a)~The procedure of entanglement testing consists of making a general two-outcome quantum measurement on the overall $n$-copy system that models the output of the device. The choice of the measurement here is arbitrary and it is precisely the experimenter's task to optimise this choice. 
(b)~Two types of error may occur: false positive, where a working device is mistaken for a faulty one, and false negative, where the opposite happens. By choosing a measurement optimally, the  probability of a false negative can be constrained to be arbitrarily small while the probability of a false positive can be made to decay exponentially fast to zero. The coefficient governing this exponential behaviour, called the Sanov exponent, is a central object of interest in this work.
}
\label{fig:entanglement_testing}
\end{figure}

There is, however, a certain freedom in choosing which type of error we quantify here. The so-called type~I error (false positive) occurs when we mistake $\rho_{AB}^{\otimes n}$ for a separable state, while a type~II error (false negative) occurs when we mistake a separable state for $\rho_{AB}^{\otimes n}$. For a fixed, arbitrarily small type~I error probability, the asymptotic exponent of the type~II error probability is known as the \emph{Stein exponent}; conversely, the asymptotic exponent of the type~I error probability with type~II probability fixed (arbitrarily small) is known as the \emph{Sanov exponent}.  
The Stein exponent of entanglement testing was first investigated in the works of Brand\~ao and Plenio~\cite{BrandaoPlenio2,Brandao2010}, although it was fully solved only very recently~\cite{gap,hayashi_stein,blurring}. Here we will focus instead on the Sanov exponent,
which we formally define as
\begin{equation}\begin{aligned}
	&\Sanov(\rho_{AB}\| \SEP_{A:B}) \\
    &\quad\ \coloneqq \lim_{\ve\to0} \liminf_{n\to\infty} - \frac1n \log_2 \min \Big\{ \Tr M_n \rho_{AB}^{\otimes n} \;\Big|\; 0 \leq M_n \leq \id,\; \Tr (\id - M_n) \sigma_n \leq \ve\ \ \forall\; \sigma_n \in \SEP_{A^n:B^n} \Big\},
\end{aligned}\end{equation}
with $\SEP_{A^n:B^n}$ standing for the set of all separable states on the $n$-partite quantum system $A^nB^n$, composed of $n$ subsystems $A^n = A_1 \ldots A_n$ held by Alice and $n$ subsystems $B^n = B_1 \ldots B_n$ held by Bob, and with $(M_n, \id - M_n)$ denoting the POVM elements of the measurement performed on the $n$-copy system. 
The evaluation of this exponent will turn out to be closely connected with the task of entanglement distillation.


\subsection{Setting of entanglement distillation}

The basic setting of entanglement distillation, as introduced in~\cite{Bennett-distillation, Bennett-distillation-mixed, Bennett-error-correction}, is as follows. 
Our protagonists, Alice and Bob, share many copies of a bipartite quantum state $\rho_{AB}$ and aim to extract pure, maximally entangled states from it. Specifically, they can apply
a sequence of quantum channels $\Lambda_n$, subjected to some 
locality constraints to be specified later, such that, when acting on $n$ copies of $\rho_{AB}$, the final state approximates $m$ copies of the maximally entangled state $\ket{\Phi_+} \coloneqq \frac{1}{\sqrt{2}} (\ket{00}+\ket{11})$, up to an error $\ve_n$. We 
write this as $\Lambda_n(\rho_{AB}^{\otimes n}) \approx_{\ve_n} \proj{\Phi_+}^{\otimes m}$, where $\ve_n$-closeness is measured by a suitable measure of distance --- either the fidelity or, equivalently, the trace distance. Crucially, although the transformation here is approximate and allows for some error,  we will require that $\lim_{n\to\infty}\ve_n = 0$: as more and more copies of $\rho_{AB}$ become available, the quality of the distilled entanglement increases, becoming perfect in the asymptotic limit. Now, if we understand $\frac{m}{n}$ as the yield of this protocol, the \emph{distillable entanglement} $E_{d}(\rho_{AB})$ is then defined as the largest asymptotic yield $\lim_{n\to\infty} \frac{m}{n}$ optimised over all feasible protocols such that the error $\ve_n$ vanishes asymptotically.

\begin{figure}[ht]
\includegraphics[width=.85\textwidth]{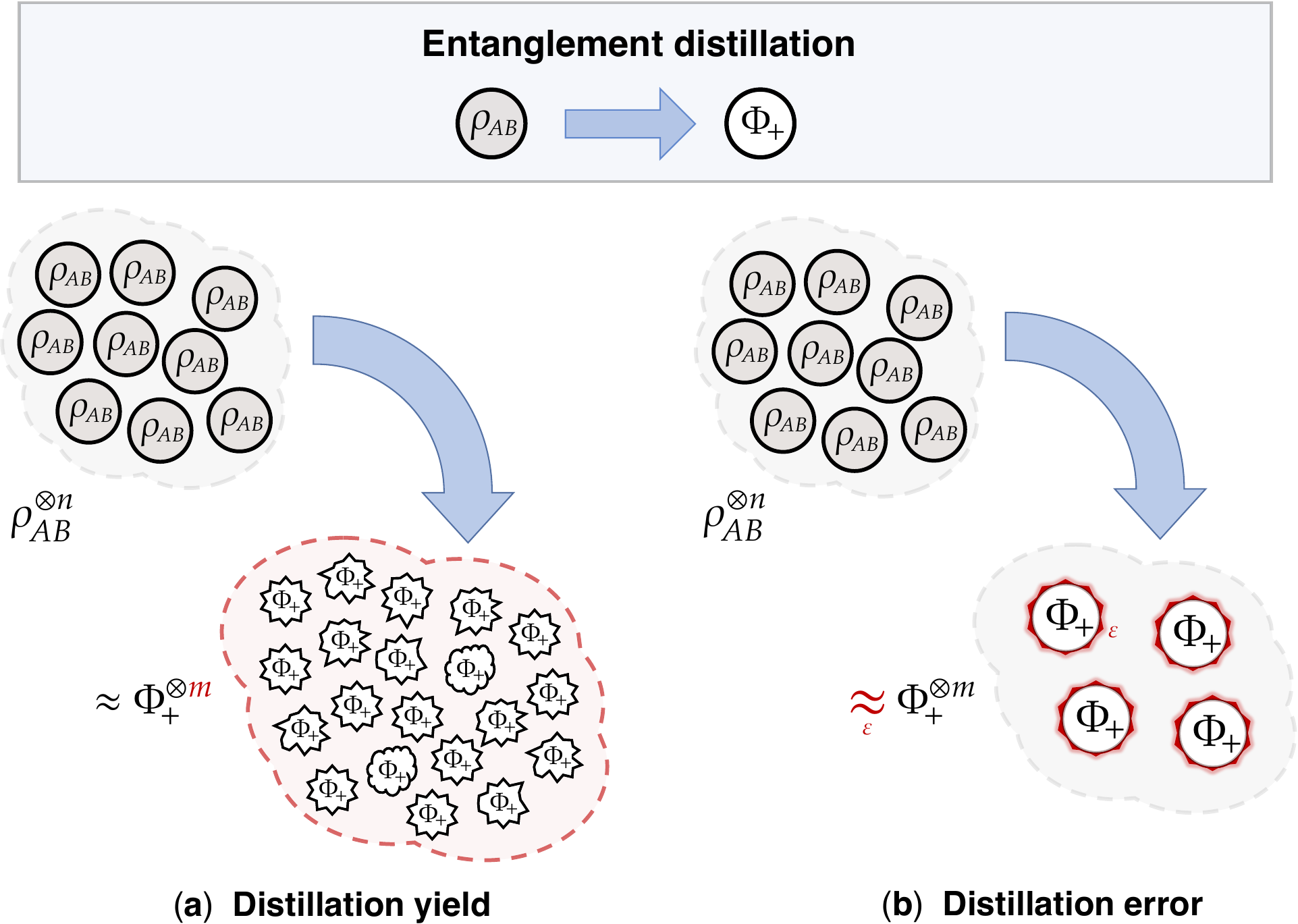}
\caption{%
\textbf{Two ways of benchmarking entanglement distillation.} Entanglement distillation is the process of converting copies of a noisy entangled quantum state $\rho_{AB}$ into fewer copies of the pure maximally entangled state $\Phi_+$. To account for physical imperfections in manipulating quantum states, the process is not required to be exact: the resulting states must approximate copies of $\Phi_+$ only to some desired degree of precision, as quantified by the distillation error $\ve$. 
(a)~Conventional approaches to distillation focus on maximising distillation \emph{yield}, that is, the number of copies of $\Phi_+$ obtained per each copy of $\rho_{AB}$. The error of the procedure is irrelevant as long as it converges to zero in the asymptotic limit as the available number of copies of $\rho_{AB}$ grows to infinity --- for a fixed number of copies, the errors may be large. 
(b)~In this paper, we instead focus on minimising the above error, potentially sacrificing some yield to obtain higher-quality entanglement. Specifically, we require that the distillation error vanish exponentially fast as the number of available copies of $\rho_{AB}$ grows, while the total number of maximally entangled states $\Phi_+$ produced in the process is still as large as desired. Accordingly, our figure of merit is not the number of copies produced but the optimal \emph{error exponent}, that is, the rate of decay of the distillation error, which directly quantifies the quality of the entanglement at the output of the protocol.
}%
\label{fig:twoways}
\end{figure}

Naturally, not all protocols $\Lambda_n$ can be implemented by two parties that are spatially separated. Therefore, the optimisation must be restricted to a suitable class of allowed protocols --- often called `free operations' --- that respect the locality constraints between Alice and Bob. While the precise choice of the free operations depends on the specific setting under consideration, the most physically natural and commonly adopted class is that of local operations and classical communication (LOCC), as defined in the original works of~\cite{Bennett-distillation, Bennett-distillation-mixed, Bennett-error-correction}. Although well motivated practically, this set is known to have an extremely complicated mathematical structure~\cite{LOCC}, hindering in particular the understanding of asymptotic entanglement transformations. This has led to a long history of alternative approaches, where one would provide additional resources or otherwise extend the allowed operations beyond the set of LOCC~\cite{Rains2001, Martin-exact-PPT, Horodecki2002, BrandaoPlenio1}, resulting in invaluable insights into the foundations of the theory as well as into the operational power of LOCC operations themselves. 
Here we follow these ideas, adopting the axiomatic framework of Brand\~ao and Plenio~\cite{BrandaoPlenio1,BrandaoPlenio2}: we consider as free all non-entangling (NE) protocols $\Lambda_n$, that is, all quantum channels which are unable to generate any entanglement --- $\Lambda_n(\sigma)$ must remain unentangled for all unentangled states $\sigma$. This weak requirement is inspired by axiomatic approaches to the second law of thermodynamics~\cite{GILES, Lieb-Yngvason}, and it has already shed light on the theory of entanglement manipulation through these fundamental thermodynamic connections~\cite{BrandaoPlenio1, irreversibility, probabilistic-reversibility, hayashi_stein, blurring}. Unlike the LOCC-based approach, the Brand\~ao--Plenio one has the added advantage of being fundamentally resource-agnostic, meaning that it can be extended beyond entanglement, leading to a unified theory of all quantum resources.

Taking inspiration from quantum hypothesis testing, where the error exponents are the figures of merit, we can apply a similar reasoning here and ask about the distillation error exponent. Specifically, consider again a distillation protocol that outputs $m$ copies of $\ket{\Phi_+}$ with error $\ve_n$. We will now ask: how fast does the quality of the distilled entanglement improve as the number of distilled copies $m$ grows to infinity? 
Instead of focusing on the optimal yield, we will thus require that $\ve_n \sim 2^{-c n}$ and characterise the optimal error exponent $c$ (see Figure~\ref{fig:twoways}). The \emph{distillable entanglement error exponent} is then defined as the largest such exponent that can be achieved as the size of the input and output systems grows:
\begin{equation}\begin{aligned}
    E_{d,\rm err} (\rho) \coloneqq \lim_{m\to\infty} \,\sup \lset \lim_{n\to\infty} - \frac1n \log_2 \ve_n \bar \Lambda_n(\rho_{AB}^{\otimes n}) \approx_{\ve_n} \proj{\Phi_+}^{\otimes m},\; \Lambda_n \in \textrm{NE} \;\, \forall\, n \rset,
\end{aligned}\end{equation}
where we 
optimise over 
sequences of non-entangling distillation protocols to find the least achievable error. One can notice that this definition no longer places any importance on the precise number of maximally entangled copies that we obtain in the protocol (provided that it can be made as large as desired), but only on the exponentially decreasing error. This provides an alternative angle 
for assessing the performance of distillation protocols, incomparable with previous approaches that focused on the distillation yield.


\subsection{Connecting entanglement testing with entanglement distillation}

A curious --- and very consequential~\cite{gap,gap-comment} --- connection between entanglement testing and distillation was shown in the works of Brand\~ao and Plenio~\cite{BrandaoPlenio2}, where the Stein exponent of entanglement testing was connected with the asymptotic yield of entanglement distillation in the axiomatic setting of non-entangling operations. Here we establish a dual to that result, proving an exact connection between the Sanov exponent and the error of entanglement distillation.

\begin{lemma}\label{lem:distillation_sanov_equivalence}
The asymptotic error exponent of entanglement distillation under non-entangling operations equals the Sanov error exponent of hypothesis testing of all separable states $\SEP_{A:B}$ against $\rho_{AB}$: 
\begin{equation}\begin{aligned}
    E_{d,\rm err} (\rho_{AB}) = \Sanov(\rho_{AB} \| \SEP_{A:B}).
\end{aligned}\end{equation}
\end{lemma}

This shows an equivalence between two a priori rather different tasks: one concerned with extracting entanglement, one with simply detecting it. The fact that the two can be so closely connected will prove extremely useful to us, as we will be able to employ the mathematical machinery of information theory to resolve the asymptotic exponent exactly. 
We stress that, although the study of entanglement distillation depends on the choice of the allowed free operations (here, non-entangling operations), the task of entanglement testing is defined independently of such constraints --- it follows the standard definition of quantum state discrimination, where all measurements allowed by quantum mechanics are considered.


\section{A generalised quantum Sanov's theorem}

Just as the exponent of hypothesis testing between two states is given by the quantum relative entropy $D(\sigma\|\rho)$, it is natural to expect the relative entropy to make an appearance in characterising the asymptotic exponent of entanglement testing. However, formalising such connections goes beyond the current state of the art in composite quantum hypothesis testing, requiring the development of new techniques.

Our main result is the complete solution of the Sanov exponent of entanglement testing, which by Lemma~\ref{lem:distillation_sanov_equivalence} also gives a resolution of the error exponent of entanglement distillation under non-entangling operations. The key role here will be played by the \emph{reverse relative entropy of entanglement}, defined as~\cite{eisert_reverse}
\begin{equation}\begin{aligned}
    D(\SEP_{A:B} \| \rho_{AB}) \coloneqq \min_{\sigma_{AB} \in \SEP_{A:B}}\, D(\sigma_{AB} \| \rho_{AB}),
\end{aligned}\end{equation}
where the term `reverse' refers to the fact that the relative entropy of entanglement was originally defined with the arguments in the opposite order, as $D(\rho_{AB}\|\SEP_{A:B}) \coloneqq \min_{\sigma_{AB} \in \SEP_{A:B}}\, D(\rho_{AB} \| \sigma_{AB})$~\cite{Vedral1997}.

\begin{thm}\label{thm:main}
For any state $\rho_{AB}$, the asymptotic Sanov error exponent of entanglement testing under all physical quantum measurements --- and, as a result, the error exponent of entanglement distillation under non-entangling operations --- equals the reverse relative entropy of entanglement:
\begin{equation}\begin{aligned}
    \Sanov(\rho_{AB} \| \SEP_{A:B}) = D(\SEP_{A:B} \| \rho_{AB}) = E_{d,\rm err} (\rho_{AB}).
\end{aligned}\end{equation}
\end{thm}

The remarkable aspect of the result is that, although both the distillable entanglement error exponent $E_{d,\rm err}$ and the Sanov exponent express \emph{asymptotic} information-theoretic properties of the quantum state $\rho_{AB}$ --- that is, they quantify the performance of $\rho_{AB}^{\otimes n}$ in the limit of large $n$ --- the quantity $D(\SEP_{A:B} \| \rho_{AB})$ is \emph{single letter}, in that it does not require a regularisation and can be evaluated on a single copy of $\rho_{AB}$.
This lets us avoid the biggest issue that plagues most solutions of asymptotic rates of entanglement manipulation protocols.

The main technical hurdle in proving Theorem~\ref{thm:main} is that the hypotheses (states) involved in the discrimination task depart from the typically considered setting of  independent and identically distributed (i.i.d.)\ ones. The recent years have seen significant interest in such hypothesis testing tasks `beyond i.i.d.'\ in quantum information theory~\cite{bjelakovic_sanov, notzel_sanov, Brandao2010, brandao_adversarial, hayashitomamichel16, Sagawa2019, berta_composite, Mosonyi2022, gap, watanabe_2024, hayashi_stein, blurring}, but none of the previous results are sufficiently general to cover our setting.
Our proof of the Theorem proceeds in two steps. First, we prove the corresponding result in the case of classical information theory --- where, instead of quantum states, we constrain ourselves to classical probability distributions. Despite the mathematically simpler structure, this result is already non-trivial, as some intuitive approaches used in i.i.d.\ cases fail to work. Instead, we employ a recently introduced powerful mathematical technique called blurring~\cite{blurring}, which allows us to handle general composite problems in hypothesis testing. Finally, we show that the classical solution can be lifted to a fully quantum one by performing suitable measurements on the quantum systems in consideration. Our solution of the problem is in fact very general and extends beyond entanglement testing to the testing of more general quantum resources. The intuition for the proof method is presented in the Methods (Section~\ref{sec:methods}), while the complete proof can be found in the \hyperlink{supp}{Supplementary Information}.

Although conceptually quite different, our result may be compared with previous related findings that evaluated asymptotic rates of entanglement distillation by connecting them with hypothesis testing problems. 
This notably includes the generalised quantum Stein's lemma, as originally conjectured in~\cite{Brandao2010} and recently proven in~\cite{hayashi_stein,blurring}. The result states that the Stein exponent of entanglement testing --- that is, the asymptotic exponent of the type~II error probability in discriminating a given state $\rho_{AB}^{\otimes n}$ from all separable states --- is given exactly by the regularised relative entropy of entanglement, $D^\infty(\rho_{AB} \| \SEP_{A:B}) \coloneqq  \lim_{n\to\infty} \frac1n D(\rho_{AB}^{\otimes n} \| \SEP_{A^n:B^n})$. As shown in~\cite{BrandaoPlenio2}, this also equals the asymptotic yield of entanglement distillation under non-entangling operations. The major difference between this result and ours is the need for regularisation: although the generalised quantum Stein's lemma ostensibly provides an exact expression of the distillable entanglement, 
this is given by a regularised quantity, which prevents an efficient evaluation of it except for limited special cases. A variant of this result was also shown in a 
setting less permissive than all non-entangling operations, namely, the more restricted class 
of `dually non-entangling operations'~\cite{DNE-distillable}, where the entanglement yield can again be evaluated through a connection with a composite hypothesis testing problem~\cite{brandao_adversarial}. 
This asymptotic rate is, however, also affected by the problem of regularisation, which our result in Theorem~\ref{thm:main} completely sidesteps.

Some words on the applicability of 
Theorem~\ref{thm:main} are in order. There are instances of quantum states from which maximal entanglement can be distilled exactly, with no error. This notably includes pure entangled states $\rho_{AB} = \proj{\psi_{AB}}$~\cite{Popescu-probabilistic}. As a consequence, in such cases the error exponent can be chosen to be arbitrarily high, and so $E_{d,\rm err}$ diverges to infinity. 
This is indeed expected and is fully captured by Theorem~\ref{thm:main}: for all such states, we have that $D(\SEP_{A:B} \| \rho_{AB}) = \infty$.
Although this looks as if it may limit the applicability of our result, such cases highly contrast with quantum states typically encountered in experimental settings:  perfect zero-error entanglement extraction is impossible from all full- or high-rank quantum states~\cite{kent_1998,fang_2020-1}, meaning that $D(\SEP_{A:B} \| \rho_{AB})$ is necessarily finite for all generic $\rho_{AB}$.
It is precisely these noisy states for which computing the asymptotic rates of entanglement distillation has been such a difficult task, as conventional techniques in entanglement distillation, which previously enabled a complete description of distillation for pure states~\cite{Bennett-distillation}, did not manage to shed much light on the general case of mixed states.
This means that our results can find direct applicability as an entanglement benchmark in the regime complementary to the well-studied and well-understood setting of noiseless pure states, serving as a well-behaved entangled measure for generic noisy quantum systems.


\section{Discussion}

The main significance of our result is the demonstration that truly asymptotic properties of entanglement can be characterised exactly without the need to consider asymptotic, regularised entanglement measures. This is important from a computational perspective --- as evaluating regularised quantities is typically extremely hard, making it difficult to quantify optimal rates and give benchmarks on feasible protocols --- but also from a theoretical one, as single-letter expressions are much easier to characterise mathematically and can lead to an improved theoretical understanding of the ultimate limitations of entanglement manipulation. 

Our findings also strengthen the connections between the theories of entanglement testing and axiomatic entanglement distillation by giving a twofold meaning to the reverse relative entropy of entanglement $D(\SEP_{A:B}\|\rho_{AB})$, an entropic entanglement measure, 
as both the optimal rate of type~I error in the task of entanglement testing as well as the error exponent of entanglement distillation under non-entangling operations.

These developments all rest on our key technical result, 
the generalised quantum Sanov's theorem. It allowed us to characterise quantum hypothesis testing tasks where one of the hypotheses is very general --- in the case of entanglement theory, it is the whole set of separable states. The result represents an advance in the theory of quantum hypothesis testing, where dealing with such non-i.i.d.\ hypotheses has long been a major obstacle. Indeed, a gap in the original proof of the generalised quantum Stein's lemma was found~\cite{gap,gap-comment}, stemming precisely from the difficulty in composite hypothesis testing; only recently have complete solutions finally appeared~\cite{hayashi_stein,blurring}, and it is precisely one of these techniques --- namely, the blurring method introduced in~\cite{blurring} --- that also allowed us to resolve the generalised quantum Sanov's theorem. 
As we discuss in more detail in Supplementary Note~\ref{app:quantum-fullproof}, the close relationship between the reverse relative entropy and quantum hypothesis testing can be extended beyond the theory of quantum entanglement to more general sets of quantum states.
Further developments of the technical methods to handle such composite, non-i.i.d.\ problems remains one of the major open problems of quantum information theory.

Our evaluation of the error exponent of entanglement distillation, on the other hand, 
provides an alternative angle that is not directly comparable with the original frameworks of entanglement distillation based on asymptotic yield~\cite{Bennett-error-correction}. Nevertheless, it is worth noting that in the latter settings, single-letter solutions were known only in very limited special cases, e.g.\ pure~\cite{Bennett-distillation} or maximally correlated states~\cite{Rains2001}. 
Additionally, simplified and computable asymptotic solutions can sometimes be obtained in `zero-error' entanglement manipulation~\cite{Martin-exact-PPT,computable-cost}, where one imposes that no error can be made whatsoever; such settings are however highly idealised and not directly useful in practice. To the best of our knowledge, our work represents the first solution of an asymptotic entanglement transformation protocol, in the sense of an asymptotic task with 
error vanishing in the limit, that admits a single-letter solution for all quantum states.

The precise connection with entanglement distillation here
relies on the choice of the axiomatic framework of non-entangling operations.
Although often considered simpler, such axiomatic approaches have not been shown to lead to single-letter expressions in the asymptotic study of entanglement before. It would certainly be interesting to extend this relation to other sets of free operations, but such an exact correspondence is most likely impossible in the most practical settings such as LOCC due to the difficulties of characterising bound entanglement~\cite{Horodecki-open-problems}. The advantage of the axiomatic approach that we have shown is that it allows for these deep connections, both conceptual and quantitative, to be established. 
Importantly, however, the equality between the reverse relative entropy and the exponent of entanglement testing is completely independent of our 
assumptions on axiomatic entanglement distillation: indeed, the task of entanglement testing does not hinge on any choice of free operations, and only uses the basic structure of quantum measurements and separable states.

A conclusion that one may draw from our approach is that, when dealing with asymptotic protocols, it can be beneficial to change one's way of looking at the problem by focusing on the error exponent rather than the asymptotic yield. This seemingly simple insight opened the door to major developments in our understanding of entanglement manipulation: it is what allowed us to establish the connection between entanglement distillation and quantum Sanov's theorem in entanglement testing, ultimately leading to a complete single-letter solution of the distillable entanglement error exponent. The basic idea can be immediately generalised to a myriad of other settings in the study of quantum and classical information, and we hope that this will lead to many more fruitful connections and developments in quantum information processing.


\section{Methods}\label{sec:methods}

The aim of this section is to provide intuition for the main technical contributions of our approach as well as the main difficulties we had to avoid on the way to establishing a generalised quantum Sanov's theorem, together with its equivalence with the exponent of entanglement distillation under non-entangling operations. Full technical proofs can be found in the \hyperlink{supp}{Supplementary Information}.


\subsection{Equivalence between entanglement distillation and entanglement testing}

Recall that our main object of study is the Sanov exponent $\Sanov(\rho_{AB}\|\SEP_{A:B})$. 
To express this exponent in a convenient way, we will use the \deff{hypothesis testing relative entropy}~\cite{Wang-Renner,Buscemi2010}
\begin{equation}\begin{aligned}
    D^\ve_H(\sigma \| \rho) \coloneqq -\log_2 \min \lset \Tr M \rho \bar 0 \leq M \leq \id,\; \Tr (\id - M) \sigma \leq \ve \rset.
\end{aligned}\end{equation}
Thinking of $M$ as an arbitrary measurement operator --- precisely, an element of a positive operator-valued measure (POVM) --- we can understand $(M, \id - M)$ as the most general two-outcome measurement that we may use to discriminate between $\rho$ and $\sigma$.
Assigning the first outcome of this measurement to the state $\sigma$ and the second to $\rho$, $D^\ve_H(\sigma\|\rho)$ then quantifies precisely the optimal type~I error exponent of hypothesis testing when the type~II error probability is constrained to be at most $\ve$. 
We can then write
\begin{equation}\begin{aligned}
    \Sanov(\rho_{AB}\| \SEP_{A:B}) \coloneqq \lim_{\ve\to0} \liminf_{n\to\infty} \frac1n D^\ve_H (\SEP_{A^n:B^n} \| \rho_{AB}^{\otimes n}),
\end{aligned}\end{equation}
where we note that the optimised hypothesis testing relative entropy can be written as
\begin{equation}\begin{aligned}
    D^\ve_H(\SEP_{A:B} \| \rho_{AB}) &= \min_{\sigma\in \SEP_{A:B}} D^\ve_H(\sigma_{AB} \| \rho_{AB}) \\
    &= -\log_2 \min \lset \Tr M \rho_{AB} \bar 0 \leq M \leq \id,\; \Tr M \sigma \geq 1-\ve \;\ \forall \sigma \in \SEP_{A:B} \rset ,
\end{aligned}\end{equation}
with the equality on the second line following from von Neumann's minimax theorem~\cite{vN-minimax}.

We remark here that the name `Sanov's theorem' is typically used in the classical information theory literature to refer to a slightly different result, concerned with the probability of observing samples whose empirical distribution (type) lies in a given set of distributions~\cite[Sec.~11.4]{CT}. We follow other works in quantum information theory, starting with~\cite{bjelakovic_sanov}, that used the name of `quantum Sanov's theorem' to refer to a hypothesis testing problem with a composite null hypothesis, similar to (but not directly comparable with) the setting we study here. We specifically use the name of `generalised quantum Sanov's theorem' as our composite hypothesis testing problem involves a general set of non-i.i.d.\ states $\SEP_{A^n:B^n}$, similar to how the `generalised quantum Stein's lemma' is now commonly used to refer to the closely related composite setting introduced in~\cite{Brandao2010}. The same generalised variant of quantum Sanov's theorem was previously studied in~\cite{hayashi_sanov}, where only bounds on the optimal asymptotic exponent were obtained (see also Section~\ref{sec:additivity} below). Yet another quantum variant of Sanov's theorem, more closely related to the original formulation of classical Sanov's theorem based on empirical distributions, was recently proposed by Hayashi~\cite{hayashi_another_sanov} --- this, however, is not directly related to the setting studied here.

The claim of our Lemma~\ref{lem:distillation_sanov_equivalence} is the asymptotic equivalence between this quantity and the error exponent of entanglement distillation, which we recall to be
\begin{equation}\begin{aligned}
    E_{d,\rm err} (\rho_{AB}) = \lim_{m\to\infty} E_{d,\rm err}^{(m)} (\rho_{AB}),
\end{aligned}\end{equation}
where we use $E_{d,\rm err}^{(m)}$ to denote the exponent of distillation under non-entangling operations for a fixed number of $m$ output copies:
\begin{equation}\begin{aligned}\label{eq:m-copy-exponent}
    E_{d,\rm err}^{(m)} (\rho_{AB}) \coloneqq \sup \lset \liminf_{n\to\infty} - \frac1n \log_2 \ve_n \bar \Lambda_n(\rho_{AB}^{\otimes n}) \approx_{\ve_n} \proj{\Phi_+}^{\otimes m},  \; \Lambda_n \in \textrm{NE} \rset,
\end{aligned}\end{equation}
where the supremum is understood to be over all sequences $(\Lambda_n)_{n\in\mathds{N}}$ of  operations satisfying the specified constraints, an in particular belonging to the class of non-entangling (NE) maps.

We now outline the main part of our argument, the details of which can be found in Supplementary Note~\ref{app:lemma1}. The approach bears some technical similarity with a construction used in~\cite{BrandaoPlenio2}, but a crucial difference is that we employ the connection in a rather different way --- instead of the type~II hypothesis testing error, which was the object of study in~\cite{BrandaoPlenio2}, we are interested in the type~I error, and suitable modifications of the proof have to be made to account for this. This shift is what distinguishes our approach and ultimately leads to quantitatively different results.

On the one hand, any distillation protocol $(\Lambda_n)_n$ can be turned into a suitable sequence of tests $(M_n, \id - M_n)$ that perform entanglement testing with a small type~I error probability. Precisely, since $\Lambda_n(\rho^{\otimes n}_{AB}) \approx_{\ve_n} \proj{\Phi_+}^{\otimes m}$, we can construct a measurement by defining $M_n \coloneqq\id - \Lambda_n^\dagger(\proj{\Phi_+}^{\otimes m}) $, where $\Lambda_n^\dagger$ denotes the adjoint map of $\Lambda_n$, representing the action of the channel in the Heisenberg picture. We can then show that the type~II error probability of this test is at most $2^{-m}$, while the type~I error is at most $\ve_n$; this gives a feasible protocol for entanglement testing, leading to the bound
\begin{equation}\label{eq:Sanov-inequality}\begin{aligned}
    \min_{\sigma_n\in \SEP_{A^n:B^n}} D^{2^{-m}}_H( \sigma_n \| \rho_{AB}^{\otimes n}) \geq \log_2 \frac{1}{\Tr M_n \rho_{AB}^{\otimes n}} \geq - \log_2 \ve_n,
\end{aligned}\end{equation}
which is one direction of the claimed relation.

For the other direction, we take any sequence of feasible measurement operators $M_n$ for entanglement testing of $\rho_{AB}^{\otimes n}$
and use them to construct a distillation protocol. This is done through a simple measure-and-prepare procedure: we first perform the measurement $(M_n, \id - M_n)$, and if we obtain the first outcome (that is, we think that the input state is separable), then we simply prepare a suitable separable state; if, however, we obtain the second outcome (that is, we think that the state is $\rho_{AB}^{\otimes n}$) then we prepare our desired target state $\ket{\Phi_+}^{\otimes m}$. In Supplementary Note~\ref{app:lemma1} we show that this constitutes a feasible distillation protocol with error $\ve_n = \Tr M_n \rho^{\otimes n}_{AB}$, giving
\begin{equation}\begin{aligned}
    E_{d,\rm err}^{(m)} (\rho_{AB}) \geq \liminf_{n\to\infty} \frac1n \min_{\sigma_n \in \SEP_{A^n:B^n}} D^{2^{-m}}_H( \sigma_n \| \rho^{\otimes n}_{AB}).
\end{aligned}\end{equation}

Altogether the above arguments show that
\begin{equation}\begin{aligned}
   E_{d,\rm err} (\rho_{AB}) = \lim_{m\to\infty} E_{d,\rm err}^{(m)} (\rho) = \lim_{m\to\infty} \liminf_{n\to\infty} \frac1n \min_{
   \sigma_n \in \SEP_{A^n:B^n}} D^{2^{-m}}_H (\sigma_{n} \| \rho_{AB}^{\otimes n}) = \Sanov(\rho_{AB} \| \SEP_{A:B})
\end{aligned}\end{equation}
which establishes an equivalence between the error exponent of entanglement distillation and the Sanov exponent of entanglement testing.


\subsection{On entropies and their (non-)additivity}\label{sec:additivity}

Let us now consider the claim of our main result, namely, that $\Sanov(\rho_{AB}\|\SEP_{A:B}) = D(\SEP_{A:B}\|\rho_{AB})$.

A simple but key observation that helps motivate this claim is that the reverse relative entropy of entanglement is in fact additive on tensor product states~\cite{eisert_reverse}. That is, we have
\begin{equation}\begin{aligned}
    D(\SEP_{AA':BB'} \| \rho_{AB} \otimes \omega_{A'B'}) = D(\SEP_{A:B} \| \rho_{AB}) + D(\SEP_{A':B'} \| \rho_{A'B'})
\end{aligned}\end{equation}
for all states $\rho_{AB}$, $\omega_{A'B'}$. To see this, let $\sigma_{AA'BB'} \in \SEP_{AA':BB'}$ be a minimiser of $D(\SEP_{AA':BB'} \| \rho_{AB} \otimes \omega_{A'B'})$, and use the fact that $\log_2(\rho_{AB} \otimes \omega_{A'B'}) = \log_2\rho_{AB} \otimes \id_{A'B'} + \id_{AB} \otimes \log_2 \omega_{A'B'}$ to get
\begin{equation}\begin{aligned}
    D(\sigma_{AA'BB'} \| \rho_{AB} \otimes \omega_{A'B'}) &= -S(\sigma_{AA'BB'}) + D(\sigma_{AB} \| \rho_{AB}) + D(\sigma_{A'B'} \| \omega_{AB}) + S(\sigma_{AB}) + S(\sigma_{A'B'})\\
&= I(AA':BB')_{\sigma} + D(\sigma_{AB} \| \rho_{AB}) + D(\sigma_{A'B'} \| \omega_{AB})\\
&\geq D(\SEP_{A:B} \| \rho_{AB}) + D(\SEP_{A':B'} \| \omega_{A'B'}),
\end{aligned}\end{equation}
where the last line follows from the non-negativity of the quantum mutual information $I(AA':BB')_{\sigma} = S(\sigma_{AB}) + S(\sigma_{A'B'}) - S(\sigma_{AA'BB'})$~\cite[Theorem~11.6.1]{MARK} and the fact that the reduced systems $\sigma_{AB}$ and $\sigma_{A'B'}$ are always separable for $\sigma_{AA'BB'}$ separable between $AA'$ versus $BB'$. This already tells us that this quantity can help us avoid issues with many-copy formulas, as regularisation is simply not needed for this formula.

While the converse direction $\Sanov(\rho_{AB}\|\SEP_{A:B}) \leq D(\SEP_{A:B}\|\rho_{AB})$ can straightforwardly be concluded from the converse of the standard i.i.d.\,setting, for the other direction, we need to construct a composite hypothesis test that works well enough to distinguish any separable state from $\rho_{AB}^{\otimes n}$. Now, consider first a simpler case: if we were to test against a fixed tensor product state $\sigma_{AB}^{\otimes n}$ instead of the whole set of separable states, the quantum Stein's lemma~\cite{Hiai1991} would immediately tell us that $D(\sigma_{AB}\|\rho_{AB})$ is an achievable error exponent. In more detail, the modern and arguably simplest approach for proving the achievability part of i.i.d.\ quantum hypothesis testing goes through the family of \deff{Petz--R\'enyi divergences} $D_{\alpha}(\sigma\|\rho) = \frac{1}{\alpha-1} \log_2 \Tr \sigma^{\alpha} \rho^{1-\alpha}$~\cite{PetzRenyi}, which leads via Audenaert et al.'s inequality~\cite{Audenaert2008} to
\begin{equation}\begin{aligned}
    \Sanov(\rho_{AB}\|\sigma_{AB}) \geq \lim_{\alpha\to 1^{-}} \lim_{n\to\infty} \frac1n D_\alpha(\sigma^{\otimes n}_{AB} \| \rho_{AB}^{\otimes n})=D(\sigma_{AB}\|\rho_{AB}).
\end{aligned}\end{equation}
Here, the crucial point in the derivation is that $D_\alpha(\sigma^{\otimes n}_{AB} \| \rho^{\otimes n}_{AB}) = n D_\alpha(\sigma_{AB} \| \rho_{AB})$ is an additive bound on the error probability that becomes asymptotically tight with $\lim_{\alpha\to 1^{-}} D_\alpha(\sigma_{AB}\|\rho_{AB}) = D(\sigma_{AB}\|\rho_{AB})$~\cite{PetzRenyi}. One might then wonder: could these state-of-the-art quantum hypothesis testing methods also be used for the generalised Sanov's theorem, where the fixed state $\sigma_{AB}^{\otimes n}$ is replaced with the set of states $\SEP_{A:B}$?

Indeed, this approach was recently initiated in~\cite{hayashi_sanov} and consequently the question was raised if the corresponding Petz--R\'enyi divergences of entanglement $D_\alpha(\SEP_{A:B} \| \rho_{AB})$ become additive. Perhaps surprisingly, however, we can show that in contrast to aforementioned special case of $\alpha=1$, the divergences are \emph{not} additive for $\alpha\in(0,1)$. Namely, by taking the antisymmetric Werner state $\rho_a$ as an example, it can be shown that~\cite{Rubboli2022} (see Supplementary Note~\ref{app:additivity-violation})
\begin{equation}\begin{aligned}
    D_\alpha(\SEP_{
    AA':BB'} \| \rho_a \otimes \rho_a) < 2 \, D_\alpha(\SEP_{A:B} \| \rho_a).
\end{aligned}\end{equation}
This non-additivity means that, in order to characterise the generalised Sanov exponent, we would really need to work with the regularised quantities $\lim_{n\to\infty} \frac1n D_\alpha(\SEP_{
A^n:B^n} \| \rho_{AB}^{\otimes n})$. Unfortunately, this prevents us from being able to use the known continuity results for the Petz--R\'enyi divergences (cf.~\cite{hayashi_sanov,gap}) and makes it difficult to follow the approach of Ref.~\cite{hayashi_sanov} to establish a connection with the reverse relative entropy $D(\SEP_{A:B} \| \rho_{AB})$, which is our goal. As such, we need to overcome this technical bottleneck in known proof techniques and develop an approach that will allow us to resolve the case of the generalised Sanov's theorem.


\subsection{Axiomatic approach} \label{sec:axioms}

Recall that our main goal is to characterise the asymptotic error exponent of the task of entanglement testing, that is, distinguishing a sequence of states $\rho^{\otimes n}_{AB}$ from the set of separable states $\SEP_{A:B}$. However, it will be useful to forget about separable states for now and try to understand the set in an axiomatic manner, using only some of its basic properties. Such an axiomatic approach is due to the influential works of Brand\~ao and Plenio~\cite{Brandao2010} in connection with the generalised quantum Stein's lemma (cf.\ the recent works~\cite{gap,hayashi_stein,blurring}).

This has a dual purpose: on the one hand, it will immediately allow us to apply many of our results to quantum resource theories beyond entanglement~\cite{Brandao-Gour,RT-review}; more importantly, however, it will actually also be a crucial ingredient in our proof of the generalised quantum Sanov's theorem for entanglement theory itself.

To this end, let us work out a list of abstract mathematical properties obeyed by the set of separable states as well as by other relevant sets of free states. 
The first five of these properties have been proposed by Brand\~{a}o and Plenio~\cite{Brandao2010}, and are sometimes known as the \deff{Brand\~{a}o--Plenio axioms}. To state them, we consider some quantum system with Hilbert space $\HH$, and a sequence $(\FF_n)_n$ of sets $\FF_n \subseteq \D\big(\HH^{\otimes n}\big)$ of density operators on $n$ copies of $\HH$. States in $\FF_n$ are conventionally referred to as \deff{free states}, and a state that is not free is called \deff{resourceful}. We posit the following:
\begin{enumerate}
\item For each $n$, $\FF_n$ is a convex and closed subset of states.
\item $\FF_1$ contains some full-rank state $\sigma_0>0$, e.g.\ the maximally mixed state.
\item The family $(\FF_n)_n$ is closed under partial traces: tracing out any number of the $n$ subsystems cannot make a free state resourceful.
\item The family $(\FF_n)_n$ is closed under tensor products: the tensor product of any two free states is also free.
\item Each $\FF_n$ is closed under permutations: permuting any of the $n$ subsystems cannot create resource from a free state.
\end{enumerate}

Picking $\HH =\HH_{AB} = \HH_A \otimes \HH_B$ as a bipartite Hilbert space and taking $\FF_n = \SEP_{A^n:B^n}$ as the set of separable states on $\HH_A^{\otimes n} \otimes \HH_B^{\otimes n}$ (with all $A$ systems on one side and all $B$ systems on the other) clearly satisfies all of the above Axioms~1--5. However, these axioms are also obeyed by many other sets of free states, corresponding to different quantum resource theories~\cite{RT-review}. 
All of our definitions can be immediately extended to such sets, with the conjectured generalised Sanov's theorem now asking about whether 
\bb
\Sanov(\rho\|\FF) \stackrel{?}{=} D(\FF\|\rho) = \min_{\sigma \in \FF_1} D(\sigma \| \rho)\, .
\ee

While the above natural set of axioms indeed turns out to be sufficient to prove the generalised quantum Stein's lemma~\cite{hayashi_stein,blurring}, it is noteworthy that the axioms are {\it not} sufficient for the generalised Sanov's theorem. In Supplementary Note~\ref{app:counterexample} we give a classical example that fulfils Axioms 1--5, while anyway having
\begin{equation}\begin{aligned}
    \Sanov(\rho\|\FF)=0<\infty=D(\FF\|\rho)
    \label{counterexample_methods}
\end{aligned}\end{equation}
for some (classical) state $\rho$.
To remedy this problem, we need to introduce a further assumption about the sets $\FF_n$. We first consider the following additional axiom:
\begin{enumerate} \setcounter{enumi}{5}
\item The \deff{regularised relative entropy of resource} is faithful, i.e.\ for all resourceful $\rho\in \D(\HH)$ with $\rho\notin \FF_1$, we have that $\rel{D^\infty}{\rho}{\FF} \coloneqq \lim_{n\to\infty} \frac1n\, \rel{D}{\rho^{\otimes n}}{\FF_n} > 0$.
\end{enumerate}

We note here that this concerns the conventional definition of the relative entropy $D^\infty(\rho\|\FF)$ rather than the `reverse' variant $D(\FF\|\rho)$. 
This rather non-trivial property is obeyed by many quantum resources encountered in practice --- for instance, in the case of separable states, it has been proved to hold independently by Brand\~{a}o and Plenio~\cite[Corollary~II.2]{Brandao2010} and by Piani~\cite{Piani2009}. It is however not universal, and the counterexample in~\eqref{counterexample_methods} violates this axiom.
Indeed, Axiom~6 turns out to be sufficient, together with Axioms~1--5, to imply the generalised Sanov theorem in the \emph{fully classical} case --- that is, when instead of general quantum states, we restrict ourselves to classical probability distributions (commuting states). However, 
the axiom does not seem to suffice to establish the quantum extension of this finding. To derive the quantum result we will instead need an 
axiom that is seemingly rather different from Axiom~6, but actually closely related to it. This new Axiom~$6^\prime$ is concerned with how measurements close to the identity act on the set of free states:

\begin{enumerate}
\item[$6^\prime$.] For some choice of numbers $r_n\in (0,1]$, the sequence $(\mathds{M}_n)_n$ of sets of measurements
\bb
\mathds{M}_n \coloneqq \left\{ \left( \frac{\id^{\otimes n} + X_n}{2},\, \frac{\id^{\otimes n} - X_n}{2}\right):\ X_n = X_n^\dag \in \LL\big(\HH^{\otimes n}\big),\ \|X_n\|_\infty \leq r_n \right\} ,
\label{axiom_6_prime}
\ee
where $\|\cdot\|_\infty$ denotes the operator norm, is \deff{compatible} with $(\FF_n)_n$~\cite{Piani2009, brandao_adversarial}. This means that whenever a measurement $\MM\in \mathds{M}_n$ is performed on the first $n$ sub-systems of a free state $\sigma\in \FF_{n+m}$, the resulting post-measurement state on the last $m$ sub-systems is also a free state in $\FF_m$, for each one of the two possible outcomes of $\MM$.
\end{enumerate}

Aside from the fact that both are obeyed by the set of separable states, as we show in Supplementary Note~\ref{app:quantum-fullproof}, it is a priori unclear why Axiom~$6^\prime$ is in any way related to Axiom~6. The connection between the two follows from the work of Piani~\cite{Piani2009}, who proved that Axiom~6 is satisfied whenever one can find a tomographically complete set of measurements that is compatible with the free states; the sets $\mathds{M}_n$ in~\eqref{axiom_6_prime} are, in fact, tomographically complete, because the POVM operators $(\id^{\otimes n} + X_n)/2$ span the space of Hermitian operators on $\HH^{\otimes n}$. It turns out that Axiom~6' is precisely what we need to prove the generalised quantum Sanov's theorem. 

In the following, our proof strategy will be to:
\begin{enumerate}[label=(\roman*)]
\item Derive the generalised Sanov's theorem for the commutative case of sets of classical states $\FF_n$ that respect Axioms 1--6 (Sections~\ref{sec:blurring} -- \ref{sec:classical-sanov}). 
\item Suitably choose measurement operations for lifting the result to the non-commutative (quantum) setting --- assuming Axioms~1--5 as well as Axiom~$6^\prime$
(Section~\ref{sec:quantum-lifting}).
\end{enumerate}


\subsection{Max-relative entropy and the blurring lemma}
\label{sec:blurring}

Instead of directly working with the hypothesis testing relative entropy $D^\ve_H(\sigma \| \rho) $, our proofs start with a dual formulation in terms of the \deff{smooth max-relative entropy}. It is defined as 
\begin{equation}\begin{aligned}
\label{eq:smooth-Dmax}
    D^\ve_{\max}(\sigma \| \rho)\coloneqq\log_2\inf\left\{\mu\in\mathds{R}\,\middle|\,\tilde{\sigma}\leq\mu\rho,\; \frac12\|\tilde{\sigma}-\sigma\|_1 \leq \ve\right\},
\end{aligned}\end{equation}
where we choose to measure $\ve$-closeness of states in terms of the trace distance. The smooth max-relative entropy enjoys, for any $\delta>0$ small enough, the duality relation~\cite{Tomamichel2013,Anshu2019,weakstrong}
\begin{equation}\begin{aligned}
D_{\max}^{\sqrt{1-\e}}(\sigma\|\rho) \leq D_H^\e(\sigma\|\rho) \leq D_{\max}^{1-\e-\delta}(\sigma\|\rho) + \log_2 \frac{1}{\delta}\, ,
\end{aligned}\end{equation}
which implies that we can essentially replace the hypothesis testing relative entropy with the smooth max-relative entropy, up to suitably modifying the smoothing parameter.

The generalised Sanov's theorem for general sets of states $\FF$ then becomes equivalent to 
\begin{equation}\begin{aligned}
\lim_{n\to\infty} \frac1n\, \rel{D_{\max}^\e}{\FF_n}{\rho^{\otimes n}} \eqt{?} D(\FF \| \rho)\qquad\forall\ \e\in (0,1),
\label{eq:sanov-dual}
\end{aligned}\end{equation}
and, using standard entropic arguments, it is not too difficult to show the special case $\ve \to 0$.
Further, since the function on the left-hand side of~\eqref{eq:sanov-dual} is monotonically non-increasing in $\e$, we immediately have that $\lim_{n\to\infty} \frac1n\, \rel{D_{\max}^\e}{\FF_n}{\rho^{\otimes n}} \leq D(\FF \| \rho)$; 
consequently, it remains to prove the opposite direction. By contradiction, our goal will be to show for the classical case that
\begin{equation}\begin{aligned}
\frac1n\, \rel{D_{\max}^\e}{\FF_n}{p^{\otimes n}} \tendsn{} \lambda < D(\FF \| p),
\label{eq:dual-approach}
\end{aligned}\end{equation}
under the assumption of Axioms 1--6.

The crucial tool for working with classical non-i.i.d.\ distributions in $\FF_n$ is the {\it blurring lemma} recently established by one of us~\cite{blurring}. Namely, for any pair of positive integers $n,m\in \N_+$ one defines the \deff{blurring map} $B_{n,m}:\R^{\XX^n} \to \R^{\XX^n}$ that transforms any input probability distribution by adding $m$ symbols of each kind $x\in \XX$, shuffling the resulting sequence, and discarding $m$ symbols.

To better understand the action of this map it is useful to recall some concepts from the theory of types~\cite{CSISZAR-KOERNER}. The \deff{type} of a sequence $x^n$ of $n$ symbols from a finite alphabet $\XX$, denoted $t_{x^n}$, is simply the empirical probability distribution of the symbols of $\XX$ found in $x^n$: in other words, $t_{x^n}(x) = \frac1n\,N(x|x^n)$, where $N(x|x^n)$ denotes the number of times $x$ appears in the sequence $x^n$. The set of $n$-types is denoted as $\mathcal{T}_n$ (we regard the alphabet as fixed). A standard counting argument reveals that the number of types is only polynomial in $n$, unlike the number of possible sequences $x^n$, which is exponential. More precisely, we have the estimate $|\mathcal{T}_n| \leq (n+1)^{|\XX|-1}$. This means that the size of the \deff{type classes}, i.e.\ the set of sequences of a given type, is generically exponential. In what follows, we will indicate with $T_{n,t}$ the type class associated with a given $n$-type $t\in \mathcal{T}_n$. Clearly, the union of all the type classes reproduces the set of all sequences.

An important observation for us is that any probability distribution $p_n$ on $\XX^n$ that is invariant under permutations, meaning that the probability of two sequences that differ only by the order of the symbols is the same, can be understood in the space of types rather than in the space of sequences. In other words, such a probability distribution is uniquely specified by the values $p_n(T_{n,t})$ it assigns to each type class. The essence of the {\it blurring lemma} as stated below in~\eqref{blurring_lemma_eq-maintext} is the analysis of the effect that the above blurring map has in type space. Since blurring perturbs the type of the input sequence a little in a random way, this action amounts to an effective `smearing' of the input probability distribution in type space: a little of the probability weight that every type class carries `spills over' to neighboring type classes.

More quantitatively, the classical one-shot blurring lemma from~\cite[Lemma~9
]{blurring} then tells us that for $\delta,\eta>0$ and $p_n, q_n \in \PP\big(\XX^n\big)$ permutationally symmetric with $p_n\left(\bigcup\nolimits_{t\in \mathcal{T}_n:\ \|s - t\|_\infty \leq \delta} T_{n,t} \right) \geq 1 - \eta$, we have
\bb
\rel{D_{\max}^{\eta}}{p_n}{B_{n,m}(q_n)} \leq \log_2 \frac{1}{q_n\left(\bigcup\nolimits_{t\in \mathcal{T}_n:\ \|s - t\|_\infty \leq \delta} T_{n,t} \right)} + n g\!\left( \left(2\delta + \tfrac1n\right) |\XX| \right),
\label{blurring_lemma_eq-maintext}
\ee
for $m=\ceil{2\delta n}$ and with the fudge function $g(x) \coloneqq (x+1)\log_2(x+1) - x\log_2 x$. We refer to Supplementary Note~\ref{app:sanov-classical} for more and further to~\cite[Lemma~9]{blurring} for a detailed technical derivation.


\subsection{Classical generalised Sanov's theorem}
\label{sec:classical-sanov}

We will now attempt to give an intuitive but mathematically non-rigorous description of the proof of the classical version of Sanov's theorem, which states that $\sanov(p\|\FF) = D(\FF \| p)$ under Axioms 1--6 of Section~\ref{sec:axioms}. Following Section~\ref{sec:blurring}, and in particular~\eqref{eq:dual-approach}, by contradiction  we can then construct two sequences of $\e$-close probability distributions $q'_n, q_n$, with $q_n \in \FF_n$, such that $q'_n \leq 2^{n\lambda} p^{\otimes n}$.

To make sense of this inequality, we have to evaluate it on a cleverly chosen set. The key tool to do that is a simple lemma by Sanov --- sometimes also known, alas, as Sanov's theorem. This tells us that~\cite{Sanov57}
\bb
p^{\otimes n} \big(\{x^n\!:\, t_{x^n}\in \AA\} \big) \leq \mathrm{poly}(n)\, 2^{-n D(\AA\|p)}
\ee
for any set of probability distributions $\AA$. It is clear what to do now: by choosing $\AA = \FF_1$, we get on the right-hand side the exponential factor $2^{n \left(\lambda - D(\FF\|p)\right)}$, which goes to zero sufficiently fast to overcome the polynomial. Thus we have that $q'_n\big(\{x^n\!:\, t_{x^n}\in \FF_1 \} \big) \ctends{}{n\to\infty}{1.5pt} 0$; in other words, a sequence drawn according to $q'_n$ has asymptotically vanishing probability of having a free type --- that is, a type in $\FF_1$. 

This at first sight may seem good, but it should make us immediately suspicious, because $q'_n$ is supposed to be $\e$-close to a \emph{free} probability distribution $q_n\in \FF_n$. It thus holds that $q_n\big(\{x^n\!:\, t_{x^n}\notin \FF_1 \} \big) \gtrsim 1-\e$ asymptotically, that is, sequences drawn with respect to the free probability distribution $q_n$ have non-free type with asymptotically non-vanishing probability. 

Let us elaborate on this intuition. Since there are only a polynomial number of types, the above reasoning shows that there exists a non-free type $s\notin 
\FF_1$ such that $q_n(T_{n,s}) \gtrsim \frac{1-\e}{\mathrm{poly}(n)}$. Of course, $s$ might depend on $n$, but for now the reader will have to trust us that up to extracting converging sub-sequences we can circumvent this obstacle (see Supplementary Note~\ref{app:sanov-classical}). So, now we have a \emph{free} probability distribution $q_n$ that has substantial (i.e.\ only polynomially vanishing) weight on a certain type class $T_{n,s}$ corresponding to a \emph{non-free} type $s\notin 
\FF_1$.

Enter blurring. By blurring $q_n$, we can make it have substantial weight not only on $T_{n,s}$ but on all type classes $T_{n,t}$ with $t\approx s$. This is precisely what blurring does: it spreads around weight among close type classes. Hence, we will have that $\tilde{q}_n(T_{n,t}) \gtrsim \frac{1-\e}{\mathrm{poly}(n)\, 2^{\alpha n}}$ for all $t\approx s$, where $\alpha >0$ is a very small exponential price we have to pay to blur $q_n$ into $\tilde{q}_n$. For a more quantitative understanding of this phenomenon we refer the reader to~\eqref{blurring_lemma_eq-maintext}, as well as to the full technical proof in Supplementary Note~\ref{app:sanov-classical}.

Now, since $\tilde{q}_n$ has substantial weight in a whole neighbourhood of types around $s$, it becomes 
ideally suited to dominate probability distributions that are very concentrated there. There is an obvious candidate for one such distribution, and it is $s^{\otimes n}$ itself! What this reasoning will eventually show is that
\bb
s^{\otimes n} \lesssim \frac{\mathrm{poly}(n)\, 2^{\alpha n}}{1-\e}\, \tilde{q}_n\, ,
\ee
where in $\lesssim$ we have swept under the carpet the fact that $s^{\otimes n}$ needs to be deprived of its exponentially vanishing non-typical tails in order for this entry-wise inequality to work.

Now we are basically done. Since blurring does not increase the max-relative entropy of resource significantly, it is possible to find a free probability distribution $r_n\in \FF_n$ such that $\tilde{q}_n \leq 2^{\beta n} r_n$ for some small $\beta>0$. Chaining the inequalities will give us
\bb
s^{\otimes n} \lesssim \frac{\mathrm{poly}(n)\, 2^{(\alpha+\beta) n}}{1-\e}\, r_n\, ,
\ee
which, by the asymptotic equipartition  property expressed as~\cite{Brandao2010,Datta-alias}
\bb
 \lim_{\ve \to 0} \liminf_{n\to\infty} \frac1n\, \rel{D^\ve_{\max}}{s^{\otimes n}}{\FF_n} = D^\infty(s\|\FF)\, ,
\ee
eventually implies that $D^\infty(s\|\FF) = 0$. This is in direct contradiction with Axiom~6, and this contradiction will complete the proof.

A full technical proof along the argument sketched above is given in Supplementary Note~\ref{app:sanov-classical}.

As a by-product of our argument, it is actually possible to design a simple explicit test that is asymptotically nearly optimal for the hypothesis task at hand. Namely, given a string of symbols $x^n\in \XX^n$ and some small tolerance $\zeta>0$:
\begin{itemize}
    \item if $\frac12 \left\|t_{x^n} - \FF_1\right\|_1 \leq \zeta$, where $t_{x^n}$ is the type of $x^n$, then we guess that the underlying probability distribution is free;
    \item otherwise, we guess that it is $p$.
\end{itemize}
This test can be shown to achieve an asymptotically vanishing type~II error probability in the limit when $n\to\infty$, and a type~I error exponent that is approximately equal to the reverse relative entropy $D(\FF\|p)$ if $\zeta>0$ is sufficiently small.


\subsection{Lifting from classical to quantum}
\label{sec:quantum-lifting}

Once a solution of the classical problem has been established, we need to extend it to quantum systems. To this end, a standard strategy is to measure: indeed, quantum measurements map quantum states to classical probability distributions, so we can use them to bring the problem to a form we can tackle with our classical result.

In the context of hypothesis testing, and, more specifically, resource testing --- where, remember, we have to distinguish between a state $\rho^{\otimes n}$ and a generic free state $\sigma_n\in \FF_n$ --- a possible strategy could be the following: we could choose a suitable measurement $\MM$ with outcomes labelled by $x\in \XX$, with $\XX$ some finite alphabet, and carry it out on every copy of the system we have been given. By doing so, we map the problem into a classical resource testing problem, where we have to distinguish between $p^{\otimes n}$, with $p \coloneqq \MM(\rho)$ being the probability distribution obtained by measuring $\rho$, and a generic free distribution $q_n \coloneqq \MM^{\otimes n}(\sigma_n)$, with $\sigma_n\in \FF_n$. 

Calling $\widetilde{\FF}_n$ the set of $q_n$'s obtained in this way, we can now try to apply the classical version of our generalised Sanov's theorem to this set. To this end, one simply needs to verify Axioms~1--6 in Section~\ref{sec:axioms} for this sequence of sets $\big(\widetilde{\FF}_n\big)_n$. While Axioms~1--5 are relatively straightforwardly checked, verifying Axiom~6 requires a more technically complex attack. We solve this problem 
by showing that Axiom~$6^\prime$ at the \emph{quantum} level directly implies Axiom~6 for the \emph{classical} sets $\widetilde{\FF}_n$ (see Theorem~\ref{quantum_Sanov_thm} in Supplementary Note~\ref{app:quantum-fullproof} for details). 
Entanglement theory satisfies also Axiom~$6^\prime$ (as proven in Corollary~\ref{cor:sep-sanov}), so we can proceed. Applying our classical generalised Sanov's theorem, we know that this strategy yields a type~I error decay equal to
\bb
\Sanov(\rho \| \FF) \geq \min_{q \in \widetilde{\FF}_1} D(q \| p) = \min_{\sigma\in \FF_1} D\big(\MM(\sigma) \,\big\|\, \MM(\rho)\big)\, .
\ee
Note that the first inequality holds because what we described is a physically possible strategy, so it yields a lower bound on the Sanov exponent. We can now further optimise over the measurement $\MM$, which yields the bound
\bb
\Sanov(\rho \| \FF) \geq \min_{\sigma\in \FF_1} D^{\mathds{ALL}}(\sigma \|\rho)\, .
\ee
Here, $D^{\mathds{ALL}}(\sigma \|\rho)$ is the measured relative entropy~\cite{Donald1986} between $\sigma$ and $\rho$, optimised over all possible measurements. 

However, we are not done yet, because the above expression is in general not equal to $\min_{\sigma\in \FF_1} D(\sigma\|\rho) = D(\FF\|\rho)$ 
due to the action of the measurement, which in general decreases the relative entropy distance between states~\cite{Berta2017}. To fix this remaining issue, we adopt a {\it double blocking} procedure. In practice, before measuring we group the $n$ systems we have at our disposal into groups of $k$ systems each (discarding the rest, if any); here, $k$ is a fixed constant. By doing so we obtain that
\bb
\Sanov(\rho \| \FF) \geq \min_{\sigma\in \FF_1} \frac1k\, D^{\mathds{ALL}}\big(\sigma_k \,\big\|\, \rho^{\otimes k}\big)\, .
\ee
Optimising over $k$ gives the main claim, because, by the entropic pinching inequality~\cite[Lemma~4.11]{TOMAMICHEL}, the right-hand side converges to $D(\FF\|\rho)$ as $k\to\infty$, as claimed. Similarly to the classical case, also in the quantum one it is possible to describe a nearly optimal test (i.e.\ a measurement) for resource testing, although in a 
less explicit way due to the involved lifting procedure.

The full proof details are given in Supplementary Note~\ref{app:quantum-fullproof}.


\section*{Acknowledgments}

We thank Hao-Chung Cheng, Filippo Girardi, Zhiwen Lin, and Roberto Rubboli for helpful discussions. MB thanks Masahito Hayashi for presenting his work~\cite{hayashi_sanov} on a composite version of quantum Sanov's theorem when visiting the Institute for Quantum Information RWTH Aachen in November 2023.

L.L. acknowledges financial support from MIUR (Ministero dell'Istruzione, dell'Universit\`{a} e della Ricerca) through the project `Dipartimenti di Eccellenza 2023--2027' of the `Classe di Scienze' department at the Scuola Normale Superiore, and from the European Union under the European Research Council (ERC Grant Agreement No.\ 101165230). M.B. acknowledges funding by the European Research Council (ERC Grant Agreement No.\ 948139) and support from the Excellence Cluster - Matter and Light for Quantum Computing (ML4Q). B.R. acknowledges the support of the Japan Society for the Promotion of Science (JSPS) KAKENHI grant no.\ 24K16984 and  Japan Science and Technology Agency (JST) PRESTO grant no.\ JPMJPR25FB. L.L. and M.B. thank RIKEN for their hospitality while part of this project was carried out.


\bibliographystyle{apsc}
\bibliography{biblio,added}


\clearpage

\appendix

\let\addcontentsline\oldaddcontentsline
\hypertarget{supp}{}
\begin{center}
{\large \textbf{--- Supplementary Information ---\\[8pt] Asymptotic quantification of entanglement with a single copy }}
\end{center}

\tableofcontents

\section{Proof of Lemma~\ref{lem:distillation_sanov_equivalence}}\label{app:lemma1}

We recall the definitions of the two relevant quantities. On the one hand, the hypothesis testing relative entropy of entanglement testing is
\begin{equation}\begin{aligned}
    D^\ve_H(\SEP_{A:B} \| \rho_{AB}) &= \min_{\sigma_{AB}\in \SEP_{A:B}} D^\ve_H(\sigma_{AB} \| \rho_{AB}) \\
    &= -\log_2 \max_{\sigma_{AB} \in \SEP_{A:B}} \min \lset \Tr M \rho_{AB} \bar 0 \leq M \leq \id,\; \Tr M \sigma \geq 1-\ve  \rset\\
    &= -\log_2 \min \lset \Tr M \rho_{AB} \bar 0 \leq M \leq \id,\; \Tr M \sigma \geq 1-\ve \; \forall \sigma \in \SEP_{A:B} \rset,
\end{aligned}\end{equation}
where the equality between the second and third line follows from von Neumann's minimax theorem~\cite{vN-minimax}. The Sanov exponent is then given by
\bb\label{eq:sanov_exp1}
\Sanov(\rho_{AB} \| \SEP_{A:B}) = \lim_{\ve\to 0} \liminf_{n\to\infty}  \frac1n D^\ve_H(\SEP_{A^n:B^n} \| \rho^{\otimes n}_{AB}).
\ee

On the other hand, the distillation exponent for $m$ copies of the maximally entangled state $\ket{\Phi_+}$~is
\begin{equation}\begin{aligned}\label{eq:m-copy-exponent_app}
    E_{d,\rm err}^{(m)} (\rho_{AB}) = \sup \lset \liminf_{n\to\infty} - \frac1n \log_2 \ve_n \bar F\big( \Lambda_n(\rho_{AB}^{\otimes n}) ,\, \proj{\Phi_m}\big) \geq 1-\ve_n, \; \Lambda_n \in \textrm{NE} \rset,
\end{aligned}\end{equation}
where  $\ket{\Phi_m} \coloneqq \ket{\Phi_+}^{\otimes m}$ and $F(\rho,\sigma) = \left\|\sqrt{\rho}{\sqrt{\vphantom{\rho}\sigma}}\right\|_1^2$ denotes the fidelity. The asymptotic error exponent of entanglement distillation is then
\bb\label{eq:sanov_exp2}
 E_{d,\rm err} (\rho_{AB}) = \lim_{m\to\infty} E_{d,\rm err}^{(m)} (\rho_{AB}).
\ee
At this point, two remarks are in order. First, we should comment about the terminology `error exponent'. Error exponents are typically defined slightly differently in traditional Shannon information theory (see \cite{Nagaoka2006} for a quantum example): namely, in common tasks one often looks for the optimal exponent of error decay under a fixed rate $r\coloneqq\frac{n}{m(n)}$, whereas we do not fix the precise number $m$ of maximally entangled copies to obtain, only requiring that $m\to\infty$ as $n\to\infty$.

Second, one could alternatively employ the trace distance $\frac12\left\|\Lambda_n(\rho_{AB}^{\otimes n}) - \proj{\Phi_m} \right\|_1$ as an error metric in Eq.~\eqref{eq:m-copy-exponent_app}. The choice between the (in)fidelity and the trace distance is immaterial in the task of distillation, and all results will be exactly the same.

{
\renewcommand{\thethm}{1}
\begin{lemma}
The asymptotic error exponent of entanglement distillation under non-entangling operations equals the Sanov error exponent of hypothesis testing of all separable states $\SEP_{A:B}$ against $\rho_{AB}$:
\begin{equation}\begin{aligned}
    E_{d,\rm err} (\rho_{AB}) = \Sanov(\rho_{AB} \| \SEP_{A:B}).
\end{aligned}\end{equation}
\end{lemma}
}
\begin{proof}
Consider any feasible distillation protocol $\Lambda_n \in \rm NE$ such that
\begin{equation}\begin{aligned}
    1-\ve_n &\leq F\big( \Lambda_n(\rho_{AB}^{\otimes n}) ,\, \proj{\Phi_m}\big)\\
    &= \bra{\Phi_m} \Lambda_n(\rho_{AB}^{\otimes n}) \ket{\Phi_m}\\
    &= \Tr \left[ \rho_{AB}^{\otimes n} \, \Lambda_n^\dagger(\proj{\Phi_m})\right]
\end{aligned}\end{equation}
where $\Lambda_n^\dagger$ denotes the adjoint map of $\Lambda_n$. Define $M_n = \id - \Lambda_n^\dagger(\proj{\Phi_m})$, which is a valid POVM element since $\Lambda_n$ is a quantum channel. Notice now that, for any 
$\sigma_n \in \SEP_{A^n:B^n}$, it holds that
\begin{equation}\begin{aligned}
    \Tr M_n \sigma_n &= 1 - \braket{\Phi_m|\Lambda_n(\sigma_n)|\Phi_m}\\
    &\geq 1 - \max_{\sigma'_n \in 
    \SEP_{A^n:B^n}} \braket{\Phi_m|\sigma'_n|\Phi_m}\\
    &\geq 1 - 2^{-m},
\end{aligned}\end{equation}
where the second line follows since $\Lambda_n$ is non-entangling, and the last line is a basic property of the maximally entangled state~\cite{Horodecki-teleportation}.
Entanglement testing with the measurement $(M_n, \id - M_n)$ thus gives a general upper bound on the distillation error exponent as
\begin{equation}\begin{aligned}
    \min_{\sigma_n \in \SEP_{A^n:B^n}} D^{2^{-m}}_H( \sigma_n \| \rho_{AB}^{\otimes n}) \geq \log_2 \frac{1}{\Tr M_n \rho_{AB}^{\otimes n}} \geq - \log_2 \ve_n.
\end{aligned}\end{equation}
Dividing by $n$, taking the supremum over $\Lambda_n\in \textrm{NE}$, and then the limit $n\to\infty$, yields
\bb \label{eq:sanov_eq1_app}
E_{d,\rm err}^{(m)}(\rho_{AB}) \leq \liminf_{n\to\infty} \frac1n \min_{\sigma_n \in \SEP_{A^n:B^n}} D^{2^{-m}}_H( \sigma_n \| \rho_{AB}^{\otimes n}) = \liminf_{n\to\infty} \frac1n\, D^{2^{-m}}_H( \SEP_{A^n:B^n} \| \rho_{AB}^{\otimes n})\, .
\ee

For the other direction, consider any measurement operator $M_n$ satisfying $\Tr M_n \sigma \geq 1-2^{-m}$ for all separable $\sigma$. 
Construct the map $\Lambda_n$ acting on the system $A^nB^n$, producing at the output an $(m+m)$-qubit operator, and defined as
\begin{equation}\begin{aligned}
    \Lambda_n (\cdot) = \Tr \!\left[ ( \id - M_n) (\cdot) \right] \proj{\Phi_m} + \Tr \left[  M_n (\cdot) \right] \frac{\id - \proj{\Phi_m}}{2^{2m}-1}\, .
\end{aligned}\end{equation}
We immediately have that $F\big( \Lambda_n(\rho_{AB}^{\otimes n}) ,\, \proj{\Phi_m}\big) = \Tr \left[ ( \id - M_n) \rho^{\otimes n} \right]$.
Furthermore, for any 
$\sigma_n \in \SEP_{A^n:B^n}$ it holds that
\begin{equation}\begin{aligned}\label{eq:isotropic}
    \Lambda_n(\sigma_n) &= 
    \Tr \!\left[ ( \id - M_n) \sigma_n \right] \left( \proj{\Phi_m} + \frac{\Tr M_n \sigma_n}{1 - \Tr M_n \sigma_n} \frac{\id - \proj{\Phi_m}}{2^{2m}-1} \right)\\
    &= 
    \Tr \!\left[ ( \id - M_n) \sigma_n \right] \left( \proj{\Phi_m} + (2^m - 1) \,\frac{\id - \proj{\Phi_m}}{2^{2m}-1}  + \eta \,\frac{\id - \proj{\Phi_m}}{2^{2m}-1} \right)
\end{aligned}\end{equation}
for some $\eta \geq 0$, since $\Tr M_n \sigma_n \geq 1-2^{-m}$ by assumption. But since both $\frac{\id - \proj{\Phi_m}}{2^{2m}-1}$ and $\proj{\Phi_m} + (2^m - 1) \,\frac{\id - \proj{\Phi_m}}{2^{2m}-1}$ are separable states~\cite{VidalTarrach}, this means that $\Lambda_n(\sigma_n)$ is also separable, and hence $\Lambda_n$ is a non-entangling map.
Thus
\begin{equation}\begin{aligned}
    E_{d,\rm err}^{(m)} (\rho_{AB}) \geq \liminf_{n\to\infty} \frac1n \log_2 \frac{1}{\Tr M_n \rho^{\otimes n}},
\end{aligned}\end{equation}
and optimising over all feasible measurements gives
\bb \label{eq:sanov_eq2_app}
E_{d,\rm err}^{(m)}(\rho_{AB}) \geq \liminf_{n\to\infty} \frac1n\, D^{2^{-m}}_H( \SEP_{A^n:B^n} \| \rho_{AB}^{\otimes n})\, .
\ee
Finally, noting that the function $D^\ve_H$ is monotonically non-decreasing in $\ve$, we have that
\begin{equation}\begin{aligned} \label{eq:sanov_eq3_app}
\lim_{\ve\to 0} \liminf_{n\to\infty} \frac1n\, 
D^{\ve}_H (\SEP_{A^n:B^n} \| \rho_{AB}^{\otimes n}) = \lim_{m\to\infty} \liminf_{n\to\infty} \frac1n\, 
D^{2^{-m}}_H (\SEP_{A^n:B^n} \| \rho_{AB}^{\otimes n})
\end{aligned}\end{equation}
since the limit $\ve \to 0$ exists and can be taken along any sequence, in particular $\ve = 2^{-m}$ with $m \to \infty$. Combining~\eqref{eq:sanov_eq1_app},~\eqref{eq:sanov_eq2_app}, and~\eqref{eq:sanov_eq3_app}, we thus infer that 
\begin{equation}\begin{aligned}
   E_{d,\rm err} (\rho_{AB}) = \lim_{m\to\infty} E_{d,\rm err}^{(m)} (\rho) = \lim_{m\to\infty} \liminf_{n\to\infty} \frac1n\, 
   D^{2^{-m}}_H (\SEP_{A^n:B^n} \| \rho_{AB}^{\otimes n}) = \Sanov(\rho_{AB} \| \SEP_{A:B})\, ,
\end{aligned}\end{equation}
concluding the proof. 
\end{proof}

\begin{rem*}
We will later see that the number of target states $m$ is actually irrelevant for the result, and in fact it holds that
\begin{equation}\begin{aligned}
    E^{(m)}_{d,\rm err}(\rho_{AB}) = E_{d,\rm err}(\rho_{AB}) = \lim_{n\to\infty} \frac1n D^\ve_H(\SEP_{A:B} \| \rho^{\otimes n}_{AB}) = D(\SEP_{A:B} \| \rho)
\end{aligned}\end{equation}
for all $m$ and all $\ve \in (0,1)$. This follows from the strong converse property of the generalised quantum Sanov's theorem (see Supplementary Note~\ref{app:quantum-fullproof}).
\end{rem*}

\begin{rem*}
The result of Lemma~\ref{lem:distillation_sanov_equivalence} holds not only in the case of separable states and non-entangling operations: the exact same proof applies to the case where the set $\SEP_{A:B}$ in Eq.~\eqref{eq:sanov_exp1} is replaced with the set $\PPT_{A:B}$ of states with positive partial transpose (PPT), and analogously non-entangling operations NE in Eq.~\eqref{eq:m-copy-exponent_app} are replaced with PPT-preserving channels, i.e.\ ones that map PPT states to PPT states. This is essentially because states of the form encountered in Eq.~\eqref{eq:isotropic}, known as isotropic states, are separable if and only if they are PPT~\cite{Horodecki-teleportation}. This will in fact allow not only for Lemma~\ref{lem:distillation_sanov_equivalence} but also our main result in Theorem~\ref{thm:main} to apply to the case where the free states are PPT states. (This immediate application of our methods to PPT states was also observed by Zhiwen Lin~\cite{zhiwen} after the initial preprint of this manuscript appeared online.)
\end{rem*}

\section{Sanov's theorem: notation and definitions}
\label{app:notation-definitions}

\subsection{Quantum setting with Brand\~{a}o--Plenio axioms} \label{sec_BP_axioms}

Let $\HH$ be a Hilbert space of finite dimension $d$. We consider a family $(\FF_n)_n$ of sets of `free states'. In practice, each $\FF_n$ is a subset of the set of density operators on $n$ copies of the system, in formula $\FF_n \subseteq \D\big(\HH^{\otimes n}\big)$. The sets of free states should satisfy some elementary properties, called the \emph{Brand\~{a}o--Plenio axioms}~\cite[p.~795]{Brandao2010}:
\begin{enumerate}
    \item Each $\FF_n$ is a convex and closed subset of $\D(\HH^{\otimes n})$, and hence also compact (since $\HH$ is finite dimensional).
    \item $\FF_1$ contains some full-rank state $\FF_1 \ni \sigma_0>0$.
    \item The family $(\FF_{n})_n$ is closed under partial traces, i.e.\ if $\sigma\in \FF_{n+1}$ then $\Tr_{n+1}\sigma \in \FF_n$, where $\Tr_{n+1}$ denotes the partial trace over the last subsystem.
    \item The family $(\FF_{n})_n$ is closed under tensor products, i.e.\ if $\sigma\in \FF_n$ and $\sigma'\in \FF_m$ then $\sigma\otimes \sigma'\in \FF_{n+m}$.
    \item Each $\FF_n$ is closed under permutations, i.e.\ if $\sigma\in \FF_n$ and $\pi\in S_n$ denotes an arbitrary permutation of a set of $n$ elements, then also $U_\pi^{\vphantom{\dag}} \sigma U_\pi^\dag\in \FF_n$, where $U_\pi$ is the unitary implementing $\pi$ over $\HH^{\otimes n}$. 
\end{enumerate}

These axioms are conceived as an abstraction of the properties of the set of separable states. But there is one more subtle property of the set of separable states that is not captured by Axioms~1--5, and it is the fact that the regularised relative entropy of entanglement is faithful on entangled states, i.e.\ it is zero on a state if and only that state is separable. This highly non-trivial fact has been proved independently by Brand\~{a}o and Plenio~\cite[Corollary~II.2]{Brandao2010} and by Piani~\cite{Piani2009}. As shown in Supplementary Note~\ref{app:counterexample}, such a conclusion does not follow solely from the above Axioms~1--5, so we will need to  introduce an additional one:

\begin{enumerate} \setcounter{enumi}{5}
\item The regularised relative entropy of resource is faithful, i.e.
\bb
\rho\notin \FF_1 \quad \Longrightarrow\quad \rel{D^\infty}{\rho}{\FF} \coloneqq \lim_{n\to\infty} \frac1n\, \rel{D}{\rho^{\otimes n}}{\FF_n} > 0\, .
\label{axiom_6}
\ee
\end{enumerate}

As usual, the existence of the limit on the right-hand side of~\eqref{axiom_6} is guaranteed by Axiom~4 (closure under tensor products) via Fekete's lemma~\cite{Fekete1923}.


\subsection{Classical setting and the method of types} \label{subsec:types}

All of the above Axioms 1 -- 6 make perfect sense for classical probability distributions, too. This is the case we are going to care about first, so we better fix some terminology. Let $\XX$ be a finite alphabet of size $|\XX|$. Denote with $\PP(\XX)$ the set of probability distributions on $\XX$. For a positive integer $n\in \N_+$, an \deff{$\boldsymbol{n}$-type} on $\XX$ is a probability distribution $t\in \PP(\XX)$ such that $nt(x)\in \N$ is an integer for all $n\in \N$. Therefore, the set of all $n$-types on $\XX$ can be defined as
\bb
\mathcal{T}_n \coloneqq \left\{ \left(\frac{k_j}{n}\right)_{j=1,\ldots, d}\!\!:\ k_j\in \N\ \ \forall\ j\, ,\ \sumno_{j=1}^d k_j = n\right\} .
\ee
In the above equation we did not indicate the dependency of $\mathcal{T}_n$ on $\XX$, as we shall always assume that $\XX$ is fixed. We will stick to this convention throughout the note in order to simplify the notation.

A well-known counting argument shows that
\bb
\left| \mathcal{T}_n \right| \leq \binom{n+d-1}{d-1}\leq (n+1)^{d-1}\, .
\label{counting_types}
\ee
For a given $t \in \mathcal{T}_n$, the associated \deff{type class} $T_{n,t}$ is the set of sequences of length $n$ made from elements in $\XX$ that have type $t$. In formula,
\bb
T_{n,t} \coloneqq \left\{ x^n\in \XX^n\!:\ N(x|x^n) = nt(x)\ \ \forall\ x\in \XX \right\} ,
\ee
where $N(x|x^n)$ denotes the number of times the symbol $x$ appears in the sequence $x^n$. Clearly, any sequence in $T_{n,t}$ can be obtained from any other such sequence by applying a suitable permutation.


\subsection{Relative entropies}

For a generic divergence $\mathds{D}$ and two sets of quantum states $\AA,\BB\subseteq \D(\HH)$, we write
\bb
\mathds{D}(\AA\|\BB) \coloneqq \inf_{a\in \AA,\, b\in \BB} D(a\|b)\, .
\label{relent_distance_sets}
\ee
In the case where $\AA = \{a\}$ is composed of one element only, we use the shorthand notation $\mathds{D}(a\|\BB)$, and similarly in the case $\BB = \{b\}$. The same definitions can be given in the case where $\AA,\BB\subseteq \PP(\XX)$ are sets of classical probability distributions.


\subsection{Filtered divergences and compatibility conditions}

Given a finite-dimensional quantum system $A$ with Hilbert space $\HH$, a \deff{measurement} on $A$ is a CPTP map $\LL\big(\HH\big) \to \R^N$, where $N\in \N^+$ is the number of outcomes (assumed to be finite). We can equivalently represent a measurement with $N$ outcomes by an $N$-outcome POVM, i.e.\ a sequence $(E_x)_{x=1,\ldots, N}$ of positive semi-definite operators $E_x\geq 0$ such that $\sum_x E_x = \id$. Given a set $\mathds{M}$ of measurements on $A$, one can use it to define a `filtered' notion of relative entropy, by setting
\bb
D^{\mathds{M}}(\rho\| \sigma) \coloneqq \sup_{\MM \in \mathds{M}} \rel{D}{\MM(\rho)}{\MM(\sigma)}\, .
\ee
A particularly simple choice of $\mathds{M}$ is the set of \emph{all} measurements on $A$, denoted by $\all$. It turns out that in this case one can relate the measured relative entropy, $D^{\all}(\rho\| \sigma)$, first considered by Donald~\cite{Donald1986} (see also~\cite{Petz-old,Berta2017}), to the standard, Umegaki relative entropy, $D(\rho\| \sigma)$. This is done via the \emph{pinching inequality}~\cite{Hayashi2002}, which gives us the relation~\cite[Lemma~4.11]{TOMAMICHEL} (see also~\cite[Eq.~(47)]{berta_composite})
\bb
D(\rho\|\sigma) - \log_2 \left|\spec(\sigma)\right| \leq D^{\all}(\rho\| \sigma) \leq D(\rho\|\sigma)\, ,
\label{entropic_pinching}
\ee
where $|\spec(\sigma)|$ denotes the number of \emph{distinct} eigenvalues of $\sigma$. Using the estimate~\cite[Eq.~(8)]{Sutter2017}
\bb
\big|\spec\big(\sigma^{\otimes k}\big)\big| \leq \binom{k+d-1}{d-1} \leq \frac{(k+d-1)^{d-1}}{(d-1)!}\, ,
\ee
where $d$ is the dimension of the underlying quantum system, one sees that for all states $\rho_k$ on $A^k$ and $\sigma$ on $A$ it holds that
\bb
\rel{D}{\rho_k}{\sigma^{\otimes k}} - \log_2 \frac{(k+d-1)^{d-1}}{(d-1)!} \leq \rel{D^{\all}}{\rho_k}{\sigma^{\otimes k}} \leq \rel{D}{\rho_k}{\sigma^{\otimes k}}\, .
\label{entropic_pinching_k_copies}
\ee
The fact that the fudge term on the left-hand side of the above inequality is of the form $\log_2 \mathrm{poly}(k)$, and therefore dividing by $k$ and taking the limit $k\to\infty$ makes it vanish, is at the heart of the \emph{asymptotic spectral pinching} method
~\cite{Hiai1991, Hayashi2002, Sutter2017}.

For any two given sets $\AA,\BB\subseteq \D(\HH)$ and an arbitrary set of measurements $\mathds{M}$, we can also plug the divergence $D^{\mathds{M}}$ into~\eqref{relent_distance_sets} and define a corresponding `filtered distance' $D^{\mathds{M}}(\AA\| \BB)$, given by
\bb
D^{\mathds{M}}(\AA\| \BB) = \inf_{\rho\in \AA,\, \sigma\in \BB} D^{\mathds{M}}(\rho\|\sigma) = \inf_{\rho\in \AA,\, \sigma\in \BB} \sup_{\MM\in \mathds{M}} \rel{D}{\MM(\rho)}{\MM(\sigma)}\, .
\label{measured_relent_distance_sets}
\ee
There are at least two natural questions to ask at this point. First, under what conditions is $D^{\mathds{M}}(\cdot\| \cdot)$ a faithful measure of statistical distance between quantum states? The key notion in this respect is that of \emph{information completeness}.

\begin{Def}
A set of measurements $\mathds{M}$ on a quantum system with finite-dimensional Hilbert space $\HH$ is said to be \deff{informationally complete} if the statistics under all measurements in $\mathds{M}$ suffice to reconstruct any state on $\HH$. In mathematical terms, this is equivalent to
\bb
\Span\left\{ E_x:\ (E_x)_x \in \mathds{M} \right\} = \LL(\HH)\, .
\ee
\end{Def}

If $\mathds{M}$ is informationally complete, then it is not difficult to see that $D^{\mathds{M}}(\cdot\| \cdot)$ is faithful over pairs of quantum states, i.e.\ $D^{\mathds{M}}(\rho\| \sigma) > 0$ whenever $\rho\neq \sigma$.

Looking again at the definition of $D^{\mathds{M}}(\AA\|\BB)$, the second natural question to ask is whether the infimum and supremum on the right-hand side of~\eqref{measured_relent_distance_sets} can be exchanged. A sufficient condition that enables precisely that has been identified by Brand\~{a}o, Harrow, Lee, and Peres in~\cite{brandao_adversarial}.

\begin{Def}[{\cite[Definition~5]{brandao_adversarial}}] \label{closed_finite_labelled_mixtures_def}
A set of measurement $\mathds{M}$ on a given quantum system is \deff{closed under finite labelled mixtures} if for all finite alphabets $\pazocal{Y}$ and all collections of POVMs $\big(E_{x|y}\big)_x\in \mathds{M}$ labelled by $y\in \pazocal{Y}$ and all probability distributions $p$ on $\pazocal{Y}$ the POVM constructed by drawing a symbol in $\pazocal{Y}$ and performing the corresponding measurement $\big(E_{x|y}\big)_x$ is still in $\mathds{M}$. In formula,
\bb
\big(p(y)\, E_{x|y}\big)_{x,y} \in \mathds{M}\, .
\ee 
\end{Def}

The following lemma from~\cite{brandao_adversarial} answers the above question.

\begin{lemma} \label{minimax_measured_relent_lemma}
Let $\AA,\BB\subseteq \D(\HH)$ be closed and convex sets of (finite-dimensional) quantum states, and let $\mathds{M}$ be a set of quantum measurements on the same system. If $\mathds{M}$ is closed under finite labelled mixtures according to Definition~\ref{closed_finite_labelled_mixtures_def}, then
\bb
D^{\mathds{M}}(\AA\| \BB) = \inf_{\rho\in \AA,\, \sigma\in \BB} \sup_{\MM\in \mathds{M}} \rel{D}{\MM(\rho)}{\MM(\sigma)} = \sup_{\MM\in \mathds{M}} \inf_{\rho\in \AA,\, \sigma\in \BB}  \rel{D}{\MM(\rho)}{\MM(\sigma)}\, ,
\ee
and both infima are achieved, i.e.\ they can be replaced by minima. 
\end{lemma}

In what follows, we will consider sequences $(\mathds{M}_n)_n$ of sets $\mathds{M}_n$ of measurements on $n$ copies of a given quantum system. When sets of restricted measurements are present alongside sets of restricted states, a key notion is that of compatibility between these two objects. We borrow the following definition from~\cite{brandao_adversarial}.

\begin{Def}[{\cite[Definition~4]{brandao_adversarial}}] \label{quantum_compatibility_def}
Let $\HH$ be a Hilbert space, $(\FF_n)_n$ a sequence of sets of quantum states $\FF_n \subseteq \D\big(\HH^{\otimes n}\big)$, and $(\mathds{M}_n)_n$ a sequence of sets $\mathds{M}_n$ of $n$-copy measurements on the corresponding quantum system. We say that $(\mathds{M}_n)_n$ is \deff{compatible} with $(\FF_n)_n$ if for all $n,m\in \N^+$, for all POVM operators $E_x^{(n)}$ that appear in a measurement belonging to $\mathds{M}_n$, and for all states $\rho_{n+m}\in \FF_{n+m}$, it holds that
\bb
\Tr_{1\ldots n} \left[ \left(E_x^{(n)}\otimes \id^{\otimes m}\right) \rho_{n+m}\right] \in \cone(\FF_m)\, ,
\ee
where $\cone (\FF_m) \coloneqq \left\{\lambda \sigma:\ \lambda\geq 0,\ \sigma \in \FF_m\right\}$.
\end{Def}

While the above definitions have all been given in the quantum setting, they can equally well be considered in the fully classical case. A {\it measurement} or {\it statistical test} on a classical system with (finite) alphabet $\XX$ is simply a classical channel $\Lambda$ with input $\XX$ and output in another discrete alphabet $\pazocal{Y}$, represented by a conditional probability distribution $\Lambda(y|x)$. Information completeness of a set of {\it classical measurements} $\mathds{L}$ simply means that
\bb
\Span \left\{ (\Lambda(y|x))_x:\ \Lambda\in \mathds{L},\ \Lambda: \R^{\XX} \to \R^{\pazocal{Y}},\ y\in \pazocal{Y} \right\} = \R^\XX ,
\label{information_completeness_classical}
\ee
where $(\Lambda(y|x))_x$ is thought of as a vector in $\R^{\XX}$. Similarly, a sequence of sets of classical measurements $(\mathds{L}_n)_n$ is compatible with a sequence of sets of classical probability distributions $(\FF_n)_n$ if for all $n,m\in \N^+$, $q_{n+m} \in \FF_{n+m}$, $\Lambda \in \mathds{L}_n$ with $\Lambda:\R^{\XX^n} \to \R^{\pazocal{Y}}$, and $y\in \pazocal{Y}$, it holds that $\tilde{q}_m \in \cone (\FF_m)$, where
\bb
\tilde{q}_m(\tilde{x}^m) \coloneqq \sum_{x^n} \Lambda(y|x^n)\, q_{n+m}(x^n\tilde{x}^m)\qquad \forall\ \tilde{x}^m\in \XX^m .
\label{classical_compatibility}
\ee

Whereas above supplementary material largely followed the literature, the methods presented in the following are completely novel.


\section{Generalised classical Sanov's theorem}
\label{app:sanov-classical}

Given a sequence of sets of free states $\FF_n$ and some $\rho\in \D(\HH)$, we consider the hypothesis testing task of distinguishing many i.i.d.\ copies of $\rho$ from a generic free state $\sigma_n \in \FF_n$. In contrast with the Stein setting, here we are interested in the rate of decay of the \emph{type I} error. The corresponding exponent, called the \deff{Sanov exponent}, is formally defined as
\bb
\sanov(\rho\|\FF) \coloneqq \lim_{\e\to 0^+} \liminf_{n\to\infty} \frac1n\, \rel{D_H^\e}{\FF_n}{\rho^{\otimes n}}\, ,
\ee
where as usual $\rel{D_H^\e}{\FF_n}{\rho^{\otimes n}} = \inf_{\sigma_n\in \FF_n} \rel{D_H^\e}{\sigma_n}{\rho^{\otimes n}}$.

The main result of this section consists in a single-letter expression for the above exponent in the classical case where $\rho = p\in \PP(\XX)$ is a probability distribution. The Sanov exponent turns out to be equal to the \deff{reverse relative entropy of resource} of $p$, given by 
\bb
D(\FF \| p) = \min_{q\in \FF_1} D(q\|p)\, . 
\ee
Our main classical result is as follows.

\begin{thm}[(Generalised classical Sanov's theorem)] \label{Sanov_thm}
Let $\XX$ be a finite alphabet, and let $(\FF_n)_n$ be a sequence of sets of probability distributions $\FF_n \subseteq \PP\big(\XX^n\big)$ that obeys Axioms~1--6 in Section~\ref{sec_BP_axioms} (i.e.\ all the Brand\~{a}o--Plenio axioms and in addition Axiom~6). Then it holds that
\bb
\lim_{n\to\infty} \frac1n\, \rel{D_H^\e}{\FF_n}{p^{\otimes n}} = D(\FF \| p)\qquad \forall\ p\in \PP(\XX),\quad \forall\ \e\in (0,1)\, ,
\label{Sanov}
\ee
entailing in particular that
\bb
\sanov(p\|\FF) = D(\FF \| p)\, .
\ee
\end{thm}


\subsection{A key tool: the blurring lemma}
\label{app:classical-blurring}

For any pair of positive integers $n,m\in \N_+$ and, as usual, a fixed alphabet $\XX$, we define the \deff{blurring map} as a linear map $B_{n,m}:\R^{\XX^n} \to \R^{\XX^n}$ that transforms any input probability distribution by adding $m$ symbols of each kind $x\in \XX$, shuffling the resulting sequence, and discarding $m$ symbols. In this way, if the input sequence is of length $n$, then the same is true of the output sequence. We can formalise the action of $B_{n,m}$ as
\bb
B_{n,m}(\cdot) \coloneqq \tr_{m} \mathcal{S}_{n+m} \left( (\cdot) \otimes \bigotimes \nolimits_x \delta_x^{\otimes m} \right) ,
\label{blurring}
\ee
where $\delta_x$ denotes the deterministic probability distribution concentrated on $x$ (i.e.\ such that $\delta_x(y) = 1$ if $y=x$, and $\delta_x(y)=0$ otherwise). Note that the output of the blurring map is always permutationally symmetric. To quantify the effect of the blurring map we use the smooth max-relative entropy, whose definition for general quantum states we recall as~\cite{RennerPhD,Datta08}
\bb
\label{eq:smooth-Dmax-SM}
    D^\ve_{\max}(\sigma \| \rho) = \log_2\inf\left\{\mu\in\mathds{R}\,\middle|\,\tilde{\sigma}\leq\mu\rho,\; \frac12\|\tilde{\sigma}-\sigma\|_1 \leq \ve\right\}.
\ee
The following is a precise statement of~\eqref{blurring_lemma_eq-maintext} from the main text, taken from~\cite{blurring}.

\begin{lemma}[{(Classical one-shot blurring~\cite[Lemma~9]{blurring})}]
\label{blurring_lemma}
Let $p_n, q_n \in \PP\big(\XX^n\big)$ be two $n$-copy probability distributions, with $p_n$ permutationally symmetric. For some $\delta,\eta>0$, let $p_n$ be $(1-\eta)$-concentrated on the $\delta$-ball of $n$-types around a single-copy probability distribution $s\in \PP(\XX)$, in the sense that
\bb
p_n\left(\bigcup\nolimits_{t\in \mathcal{T}_n:\ \|s - t\|_\infty \leq \delta} T_{n,t} \right) \geq 1 - \eta\, ,
\ee
where $\|s - t\|_\infty \coloneqq \max_{x\in \XX} \left|s(x) - t(x)\right|$. Then for $m=\ceil{2\delta n}$ it holds that
\bb
\rel{D_{\max}^{\eta}}{p_n}{B_{n,m}(q_n)} \leq \log_2 \frac{1}{q_n\left(\bigcup\nolimits_{t\in \mathcal{T}_n:\ \|s - t\|_\infty \leq \delta} T_{n,t} \right)} + n g\!\left( \left(2\delta + \tfrac1n\right) |\XX| \right) ,
\label{blurring_lemma_eq}
\ee
where the blurring map $B_{n,m}$ is defined by~\eqref{blurring}, and the function $g$ is the `bosonic entropy function'
\bb
g(x) \coloneqq (x+1)\log_2(x+1) - x\log_2 x.
\label{g}
\ee
\end{lemma}

Note that if $q_n\big(\bigcup\nolimits_{t\in \mathcal{T}_n:\ \|s - t\|_\infty \leq \delta} T_{n,t} \big) = 0$ then~\eqref{blurring_lemma_eq} holds trivially with the convention that $\log_2 1/0 = \infty$.


\subsection{Proof of the generalised Sanov's theorem: easy part}

Throughout this subsection we will be analysing the fully quantum case. In the next one, instead, we will restrict ourselves to classical probability distributions. 

Because of the weak/strong converse duality~\cite{Tomamichel2013,Anshu2019,weakstrong}
\bb
D_{\max}^{\sqrt{1-\e}}(\rho\|\sigma) \leq D_H^\e(\rho\|\sigma) \leq D_{\max}^{1-\e-\delta}(\rho\|\sigma) + \log_2 \frac{1}{\delta},
\label{weak_strong_converse_duality}
\ee
which holds for all $\e \in (0,1)$ and all $\delta \in (0,1-\ve]$, 
the generalised Sanov's theorem is equivalent to
\bb
\lim_{n\to\infty} \frac1n\, \rel{D_{\max}^\e}{\FF_n}{\rho^{\otimes n}} \eqt{?} D(\FF \| \rho)\qquad \forall\ \rho\in \D(\HH),\quad \forall\ \e\in (0,1)\, .
\ee
We start by establishing that the above relation holds when $\e\to 0^+$. We first need a preliminary lemma.

\begin{lemma} \label{e_relent_convergence_lemma}
Let $\FF\subseteq \D(\HH)$ be a closed subset of states. For some $\rho\in \D(\HH)$ and $\e>0$, define $D^\e(\FF\|\rho) \coloneqq \min_{\sigma':\, \frac12 \|\sigma' - \FF\|_1\leq \e} D(\sigma'\|\rho)$, where $\|\sigma' - \FF\|_1 \coloneqq \min_{\sigma\in \FF} \|\sigma'- \sigma\|_1$. Then it holds that
\bb
\lim_{\e\to 0^+} D^\e(\FF\|\rho) = D(\FF\|\rho)\, .
\ee
\end{lemma}

\begin{proof}
Clearly, $D^\e(\FF\|\rho) \leq D(\FF\|\rho)$ for all $\e>0$. Therefore, we only have to prove that $\liminf_{\e\to 0^+} D^\e(\FF\|\rho) \geq D(\FF\|\rho)$. Let $(\e_n)_n$ be a sequence of numbers $\e_n>0$ such that $\lim_{n\to\infty} \e_n = 0$ and $\lim_{n\to\infty} D^{\e_n}(\FF\|\rho) = \liminf_{\e\to 0^+} D^\e(\FF\|\rho)$. For all $n$, choose a point of minimum $\sigma'_n \in \D(\HH)$ of the lower semi-continuous function $D(\cdot\,\|\,\rho)$ on the closed set $\{\sigma':\, \frac12 \|\sigma'-\FF\|_1\leq \e_n\}$. Up to taking sub-sequences, we can assume that $\lim_{n\to\infty} \sigma'_n = \sigma$ converges. Note that since $\FF$ is closed, it must be that $\sigma \in \FF$. Now, using the lower semi-continuity of the relative entropy we have that
\bb
\liminf_{\e\to 0^+} D^\e(\FF\|\rho) = \lim_{n\to\infty} D^{\e_n}(\FF\|\rho) = \lim_{n\to\infty} D(\sigma'_n\|\rho) \geq D(\sigma\|\rho) \geq D(\FF\|\rho)\, ,
\ee
which concludes the proof.
\end{proof}

\begin{prop} \label{limit_e_to_0_prop}
Let $(\FF_n)_n$ be a sequence of sets of states $\FF_n \subseteq \D\big(\HH^{\otimes n}\big)$ that obeys the Brand\~{a}o--Plenio axioms (Axioms~1--5 in Section~\ref{sec_BP_axioms}). Then for all states $\rho\in \D(\HH)$ it holds that
\bb
\lim_{\e \to 0^+} \liminf_{n\to\infty} \frac1n\, \rel{D_{\max}^\e}{\FF_n}{\rho^{\otimes n}} = \lim_{\e \to 0^+} \limsup_{n\to\infty} \frac1n\, \rel{D_{\max}^\e}{\FF_n}{\rho^{\otimes n}} = D(\FF \| \rho)\, .
\label{limit_e_to_0}
\ee
\end{prop}

\begin{proof}
Fix $\sigma \in \FF_1$. Due to Axiom~4, we have that $\sigma^{\otimes n}\in \FF_n$ for all $n\in \N^+$. Therefore
\bb
\lim_{\e \to 0^+} \limsup_{n\to\infty} \frac1n\, \rel{D_{\max}^\e}{\FF_n}{\rho^{\otimes n}} &\leq \lim_{\e \to 0^+} \limsup_{n\to\infty} \frac1n\, \rel{D_{\max}^\e}{\sigma^{\otimes n}}{\rho^{\otimes n}} \\
&= D(\sigma\|\rho)\, ,
\ee
where the last step is by the asymptotic equipartition property in the i.i.d.\ case~\cite{tomamichel_2009,TOMAMICHEL}.
Minimising over $\sigma\in \FF_1$, we obtain that
\bb
\lim_{\e \to 0^+} \limsup_{n\to\infty} \frac1n\, \rel{D_{\max}^\e}{\FF_n}{\rho^{\otimes n}} \leq D(\FF\|\rho)\, .
\ee
For the other direction, we have that
\bb
\rel{D_{\max}^\e}{\FF_n}{\rho^{\otimes n}} &\geq \rel{D^\e}{\FF_n}{\rho^{\otimes n}} \\
&= \min_{\sigma'_n:\ \frac12 \|\sigma'_n - \FF_n\|_1\leq \e} \rel{D}{\sigma'_n}{\rho^{\otimes n}} \\
&\geqt{(i)} \min_{\sigma'_n:\ \frac12 \|\sigma'_n - \FF_n\|_1\leq \e} \sum_{j=1}^n \rel{D}{\sigma'_{n,j}}{\rho} \\
&\geqt{(ii)} n\, D^\e(\FF\|\rho)\, .
\label{limit_e_to_0_proof_key_inequality}
\ee
Here, in~(i) we denoted with $\sigma'_{n,j}$ the reduced state of $\sigma'_n$ on the $j^\text{th}$ subsystem. To justify~(i), note that by the sub-additivity of entropy
\bb
S(\sigma'_n) \leq \sum_{j=1}^n S\big(\sigma'_{n,j}\big)\, ,
\ee
which in turn implies that
\bb
\rel{D}{\sigma'_n}{\rho^{\otimes n}} = - S(\sigma'_n) + \sum_{j=1}^n \left( S\big(\sigma'_{n,j}\big) + \rel{D}{\sigma'_{n,j}}{\rho} \right) \geq \sum_{j=1}^n \rel{D}{\sigma'_{n,j}}{\rho}\, .
\ee
In~(ii), instead, we simply observed that due to data processing and Axiom~3 it holds that
\bb
\frac12 \,\big\| \sigma'_{n,j} - \FF_1\big\|_1 \leq \frac12 \, \big\|\sigma'_n - \FF_n\big\|_1 \leq \e
\ee
for all $j=1,\ldots, n$. To conclude the proof, divide both sides of~\eqref{limit_e_to_0_proof_key_inequality} by $n$, take the limit inferior as $n\to\infty$ first and $\e\to 0^+$ second, and use Lemma~\ref{e_relent_convergence_lemma}.
\end{proof}


\subsection{Proof of the generalised Sanov's theorem: hard part} \label{subsec:classical_Sanov_hard_part}

As anticipated, we now restrict ourselves to the classical case. At this point we know that
\bb
\lim_{n\to\infty} \frac1n\, \rel{D_{\max}^\e}{\FF_n}{p^{\otimes n}} \leq D(\FF \| p)\qquad \forall\ p\in \PP(\XX),\quad \forall\ \e\in (0,1)\, ,
\ee
simply because the function on the left-hand side is monotonically non-increasing in $\e$, and its limit for $\e\to 0^+$ corresponds to the right-hand side due to Proposition~\ref{limit_e_to_0_prop}. To conclude, we need to prove the converse inequality. Before we start, let us recall a useful estimate due to Sanov himself, which, unfortunately, is sometimes also known as Sanov's theorem \cite{Sanov57,CT}. See also the particularly clear formulation in~\cite[Exercise~2.12, p.~29]{CSISZAR-KOERNER}.

\begin{lemma} \label{Sanov_lemma}
Let $\AA\subseteq \PP(\XX)$ be a subset of probability distributions over a finite alphabet $\XX$. Then for all $p\in \PP(\XX)$ and for all $n\in \N^+$ it holds that
\bb
p^{\otimes n} \big(\{x^n\!:\, t_{x^n}\in \AA\} \big) = p^{\otimes n} \left(\bigcup\nolimits_{t\in \mathcal{T}_n\cap \AA} T_{n,t} \right) \leq (n+1)^{|\XX| - 1}\, 2^{-n D(\AA\|p)} .
\ee
If $\AA$ is convex, then the polynomial factor $(n+1)^{|\XX| - 1}$ in the rightmost side can be omitted.
\end{lemma}

We are now ready to provide the complete proof of the generalised classical Sanov's theorem (Theorem~\ref{Sanov_thm}).

\begin{proof}[Proof of Theorem~\ref{Sanov_thm}]
As we just discussed, it suffices to show that
\bb
\liminf_{n\to\infty} \frac1n\, \rel{D_{\max}^\e}{\FF_n}{p^{\otimes n}} \geqt{?} D(\FF \| p)\qquad \forall\ p\in \PP(\XX),\quad \forall\ \e\in (0,1)\, .
\label{proof_Sanov_eq1}
\ee
We will proceed by contradiction. Assume that for infinitely many values of $n$ there exists 
\bb
q'_n\in \PP(\XX^n)\, , \quad q_n \in \FF_n\, ,\quad \lambda < D(\FF\|p)\, ,
\label{proof_Sanov_eq2}
\ee
such that
\bb
\frac12 \,\big\|q_n - q'_n\big\|_1\leq \e\, ,\quad q'_n \leq 2^{n\lambda} p^{\otimes n} ,
\label{proof_Sanov_eq3}
\ee
where in the second expression $\leq$ denotes entry-wise inequality between vectors in $\R^{\XX^n}$. For the rest of this proof, unless otherwise specified, every limit $n\to\infty$ is to be intended as taken on the diverging sequence of values of $n$ for which~\eqref{proof_Sanov_eq2} -- \eqref{proof_Sanov_eq3} hold.

Due to Lemma~\ref{e_relent_convergence_lemma}, we can find some $\zeta > 0$ with the property that
\bb
D^\zeta(\FF\|p) > \lambda\, .
\label{proof_Sanov_eq4}
\ee
Evaluating~\eqref{proof_Sanov_eq3} on the set 
\bb
\pazocal{Y}_{n,\zeta} \coloneqq \left\{x^n\!:\, \frac12\, \big\|t_{x^n} - \FF_1\big\|_1\leq \zeta \right\} = \bigcup_{t\in\mathcal{T}_n:\, \frac12 \|t - \FF_1\|_1\leq \zeta} T_{n,t}
\label{proof_Sanov_eq5}
\ee
of sequences whose type is $\zeta$-close to a free probability distribution yields
\bb
q'_n \big(\pazocal{Y}_{n,\zeta}\big) &\leq 2^{n\lambda}\, p^{\otimes n}\big(\pazocal{Y}_{n,\zeta}\big) \\
&\leqt{(i)} 2^{n\lambda}\, (n+1)^{|\XX|-1}\, 2^{-n D^\zeta(\FF\|p)} \\
&\hspace{-7pt} \tendsn{\text{(ii)}} 0\, ,
\label{proof_Sanov_eq6}
\ee
where in~(i) we used Lemma~\ref{Sanov_lemma}, and~(ii) follows from~\eqref{proof_Sanov_eq4}. 

We thus see that
\bb
q'_n \big(\pazocal{Y}_{n,\zeta}^c\big) \tendsn{} 1\, ,
\ee
entailing that
\bb
\liminf_{n\to\infty} q_n\big(\pazocal{Y}_{n,\zeta}^c\big) \geq 1-\e
\label{proof_Sanov_eq8}
\ee
because of~\eqref{proof_Sanov_eq3}. In particular,
\bb
q_n\big(\pazocal{Y}_{n,\zeta}^c\big) \geq \frac{1-\e}{2}
\ee
for all sufficiently large $n$ in the sequence. We further have
\bb
q_n\big(\pazocal{Y}_{n,\zeta}^c\big) = q_n\left(\bigcup\nolimits_{s\in\mathcal{T}_n:\, \frac12 \|s - \FF_1\|_1 > \zeta} T_{n,s}\right) = \sum_{s\in\mathcal{T}_n:\, \frac12 \|s - \FF_1\|_1 > \zeta} q_n(T_{n,s})
\ee
as well as $|\mathcal{T}_n| \leq (n+1)^{|\XX|-1}$, and as such we can pick for all $n$ in the sequence some $s_n\in\mathcal{T}_n$ such that $\frac12 \|s_n - \FF_1\|_1 > \zeta$ and
\bb
q_n(T_{n,s_n}) \geq \frac{1-\e}{2\, (n+1)^{|\XX|-1}}\, .
\label{proof_Sanov_eq11}
\ee
By compactness, up to extracting a sub-sequence we can assume that $s_n \ctends{}{n\to\infty}{1.5pt} s$ for some $s \in \PP(\XX)$. For the rest of this proof, unless otherwise specified, every limit $n\to\infty$ is understood to be taken on this sub-sequence. Note that $\frac12\, \big\| s - \FF_1\big\|_1\geq \zeta$, implying in particular that
\bb
s\notin \FF_1\, .
\label{proof_Sanov_eq12}
\ee

Now, fix two arbitrarily small $\delta,\eta>0$. By typicality,
\bb
s^{\otimes n} \left(\bigcup\nolimits_{t \in \mathcal{T}_n:\ \|s - t\|_\infty \leq \delta} T_{n,t} \right) \geq 1 - \eta
\ee
holds for all sufficiently large $n$. Therefore, the blurring lemma ensures that
\bb
\rel{D^\eta_{\max}}{s^{\otimes n}}{B_{n,m}(q_n)} &\leqt{(iii)} - \log_2 q_n \left(\bigcup\nolimits_{t \in \mathcal{T}_n:\ \|s - t\|_\infty \leq \delta} T_{n,t} \right) + n g\big(\big(2\delta + \tfrac1n\big) |\XX|\big) \\
&\leqt{(iv)} - \log_2 q_n (T_{n,s_n}) + n g\big(\big(2\delta + \tfrac1n\big) |\XX|\big) \\
&\leqt{(v)} -\log_2(1-\e) + 1 + \left(|\XX|-1\right) \log_2(n+1) + n g\big(\big(2\delta + \tfrac1n\big) |\XX|\big)\, ,
\label{proof_Sanov_eq14}
\ee
where $m = \ceil{2\delta n}$. The justification of the above chain of inequalities is as follows: in~(iii) we applied Lemma~\ref{blurring_lemma} with $p_n = s^{\otimes n}$ (which is clearly permutationally symmetric); (iv)~holds for all sufficiently large $n$ in the sub-sequence, because $s_n \ctends{}{n\to\infty}{1.5pt} s$; finally, in~(v) we leveraged the estimate in~\eqref{proof_Sanov_eq11}.

We are almost done. To conclude, it suffices to note that by Axioms~1,~2,~4, and~5 together with standard properties of the max-relative entropy from~\eqref{eq:smooth-Dmax-SM} with $\ve=0$, it holds that
\bb
\rel{D_{\max}}{B_{n,m}(q_n)}{\FF_n} \leq m |\XX| \log_2 \frac1\mu = \ceil{2\delta n} |\XX| \log_2 \frac1\mu\, ,
\label{proof_Sanov_eq15}
\ee
where $\mu \coloneqq \max_{q_0 \in \FF_1} \min_{x\in \XX} q_0(x)$, which is strictly positive by Axiom~2. Therefore, combining~\eqref{proof_Sanov_eq14} and~\eqref{proof_Sanov_eq15} thanks to the triangle inequality for the smooth max-relative entropy \cite{Matthias-Alex}, we see that
\bb
\rel{D^\eta_{\max}}{s^{\otimes n}}{\FF_n} &\leq \rel{D^\eta_{\max}}{s^{\otimes n}}{B_{n,m}(q_n)} + \rel{D^\eta_{\max}}{B_{n,m}(q_n)}{\FF_n} \\
&\leq -\log_2(1-\e) + 1 + \left(|\XX|-1\right) \log_2(n+1) + n g\big(\big(2\delta + \tfrac1n\big) |\XX|\big) + \ceil{2\delta n} |\XX| \log_2 \frac1\mu\, .
\ee
We can now divide by $n$ and take $n\to\infty$, which yields
\bb
{\limsup_{n\to\infty}}' \frac1n\, \rel{D^\eta_{\max}}{s^{\otimes n}}{\FF_n} \leq g\big(2\delta |\XX|\big) + 2\delta |\XX| \log_2 \frac1\mu\, ,
\ee
where we denoted with $\limsup'_{n\to\infty}$ the limit on the sub-sequence. Taking into account all values of $n\in \N^+$ instead of a sub-sequence only, the above reasoning shows that
\bb
\liminf_{n\to\infty} \frac1n\, \rel{D^\eta_{\max}}{s^{\otimes n}}{\FF_n} \leq g\big(2\delta |\XX|\big) + 2\delta |\XX| \log_2 \frac1\mu\, .
\ee
Taking the limits $\eta \to 0^+$ and $\delta \to 0^+$ shows that
\bb
D^\infty(s\|\FF) \leqt{(vi)} \lim_{\eta \to 0^+} \liminf_{n\to\infty} \frac1n\, \rel{D^\eta_{\max}}{s^{\otimes n}}{\FF_n} \leq 0\, ,
\ee
where~(vi) is from asymptotic continuity. (Incidentally, Brand\~{a}o--Plenio--Datta's asymptotic equipartition property~\cite{Brandao2010,Datta-alias} guarantees that~(vi) is actually an equality, although we are not using this particular fact here.) Since $D^\infty(s\|\FF)\geq 0$ holds by construction, we have just proved that
\bb
D^\infty(s\|\FF) = 0\, ,\quad s\notin \FF_1\, ,
\ee
which is in blatant contradiction with Axiom~6.
\end{proof}


\subsection{On the optimal test in the classical case} \label{subsec:optimal_test_classical}

The difficult part of the proof of the generalised classical Sanov's theorem (Theorem~\ref{Sanov_thm}), which has been presented in the above Section~\ref{subsec:classical_Sanov_hard_part}, consists, in technical terms, in an achievability statement. Namely, by proving~\eqref{proof_Sanov_eq1} we have indirectly shown that there exists a tests that can distinguish $p^{\otimes n}$ from an arbitrary free probability distribution with asymptotically vanishing type II error and type I error exponent arbitrarily close to the reverse relative entropy $D(\FF\|p)$. This naturally begs the question: how can such a test be designed in practice? Here we would like to argue that a test with this property can be described as follows. Given a string of alphabet symbols $x^n\in \XX^n$, we start by calculating its type $t_{x^n}$, a probability distribution on $\XX$ given by $t_{x^n}(x) \coloneqq \frac{N(x|x^n)}{n}$, where $N(x|x^n)$ is the number of times $x$ appears in $x^n$ (see Section~\ref{subsec:types}). Then, for some small tolerance $\zeta>0$:
\begin{itemize}
    \item if $\frac12 \left\|t_{x^n} - \FF_1\right\|_1 \leq \zeta$, then we guess that the underlying probability distribution is free;
    \item otherwise, we guess that it is $p$.
\end{itemize}
Using the notation introduced in Eq.~\eqref{proof_Sanov_eq5}, an alternative way of rephrasing the above test is by saying that we guess $p$ if and only if $x^n\notin \pazocal{Y}_{n,\zeta}$.

The type II error probability induced by the above test, i.e.\ the (worst-case) probability that a free probability distribution is mistakenly identified as $p$, is given by
\bb
\beta_n \coloneqq \max_{q_n\in \FF_n} \sum_{x^n\notin \pazocal{Y}_{n,\zeta}} q_n(x^n) = \max_{q_n\in \FF_n} q_n\big(\pazocal{Y}_{n,\zeta}^c\big) .
\ee
The above proof shows that assuming Eq.~\eqref{proof_Sanov_eq8} for any diverging subsequence of values of $n$ and for any $\e\in (0,1)$ leads to a contradiction, which is equivalent to stating that
\bb
\lim_{n\to\infty} \beta_n = 0\, .
\ee
That is, the type II error probability vanishes asymptotically, which matches our first requirement for an optimal test.

What about the type I error exponent? The type I error probability $\alpha_n$ satisfies that
\bb
\alpha_n = \sum_{x^n\in \pazocal{Y}_{n,\zeta}} p^{\otimes n}(x^n) = p^{\otimes n}\big(\pazocal{Y}_{n,\zeta}\big) \leq 2^{-n D^\zeta(\FF\|p)} ,
\label{optimal_test_classical_type_I}
\ee
where the last inequality follows from Lemma~\ref{Sanov_lemma}, once one realises that the set $\big\{q\in \PP(\XX):\, \frac12 \|q-\FF_1\|_1 \leq \zeta\big\}$ is convex; here,
\bb
D^\zeta(\FF\|p) \coloneqq \min_{q:\, \frac12 \|q-\FF_1\|_1 \leq \zeta} D(q\|p)\, .
\ee
The relation~\eqref{optimal_test_classical_type_I} shows directly that the type I error exponent is at least equal to $D^\zeta(\FF\|p)$. By Lemma~\ref{e_relent_convergence_lemma}, $\lim_{\zeta\to 0^+} D^\zeta(\FF\|p) = D(\FF\|p)$; hence, by choosing $\zeta>0$ sufficiently small, the exponent can be made arbitrarily close to the reverse relative entropy $D(\FF\|p)$, meeting also the second requirement on the optimal test.


\subsection{How to verify Axiom~6}

Axiom~6 may seem like a strange one. For a start, it involves directly an entropic quantity, a feature that appears to be in contradiction with the philosophy of information theory, that prescribes that entropic quantities should find their meaning operationally rather than axiomatically. But even more worryingly, it may seem very hard to verify Axiom~6 for any given sequence of sets of probability distributions $(\FF_n)_n$, as doing so would involve estimating a regularised quantity. To solve this issue, we illustrate here a simple sufficient condition that allows a swift verification of Axiom~6. The result below was proved by Piani in the fully quantum case in a pioneering work~\cite{Piani2009}. Here we need only its classical version, which we state here for the sake of completeness.

\begin{lemma}[{\cite[Theorem~1 and ensuing discussion]{Piani2009}}] \label{Piani_lemma_to_verify_axiom_6}
For a finite alphabet $\XX$, let $(\FF_n)_n$ be a sequence of sets of probability distributions $\FF_n\subseteq \PP(\XX^n)$, and let $\mathds{L} = (\mathds{L}_n)_n$ be a sequence of sets of classical channels $\mathds{L}_n$ on $\XX^n$. If the compatibility condition in Definition~\ref{quantum_compatibility_def} (see also~\eqref{classical_compatibility} for a classical formulation) is obeyed, then
\bb
D^\infty(p\|\FF) \geq D^{\mathds{L}}(p\|\FF)\, ;
\ee
if, in addition, $\mathds{L}_1$ is informationally complete and $\FF$ is topologically closed, then $D^\infty(\cdot\|\FF)$ if faithful, in the sense that
\bb
D^\infty(p\|\FF) > 0\qquad \forall\ p \notin \FF\, .
\ee
\end{lemma}


\section{Generalised quantum Sanov's theorem}
\label{app:quantum-fullproof}

To prove the result in the quantum case, we consider another axiom, concerned this time with the existence of a suitable informationally complete and compatible sequence of measurements. We will use this axiom to guarantee that Axiom~6 is satisfied for probability distributions resulting from this sequence, allowing us to apply the classical generalised Sanov's theorem (Theorem~\ref{Sanov_thm}) and lift the classical result to the quantum case. Specifically, consider the following addition to the set of axioms studied in Section~\ref{sec_BP_axioms}.
\begin{enumerate}
\item[$6^\prime$.] For some choice of numbers $r_n\in (0,1]$, the sequence $(\mathds{M}_n)_n$ of sets of measurements
\bb
\mathds{M}_n \coloneqq \left\{ \left( \frac{\id^{\otimes n} + X_n}{2},\, \frac{\id^{\otimes n} - X_n}{2}\right):\ X_n = X_n^\dag \in \LL\big(\HH^{\otimes n}\big),\ \|X_n\|_\infty \leq r_n \right\} ,
\label{quantum_Sanov_thm_compatible_set}
\ee
where $\|\cdot\|_\infty$ denotes the operator norm, is compatible with $(\FF_n)_n$.
\label{axiom_6prime}
\end{enumerate}

We are now ready to state our main result in its most general form.

\begin{thm}[(Generalised quantum Sanov's theorem)] \label{quantum_Sanov_thm}
Let $\HH$ be a finite-dimensional Hilbert space, and let $(\FF_n)_n$ be a sequence of sets of quantum states $\FF_n \subseteq \D\big(\HH^{\otimes n}\big)$ that obeys the Brand\~{a}o--Plenio axioms (Axioms~1--5 in Section~\ref{sec_BP_axioms}) as well as Axiom~$6^\prime$. Then, we have that
\bb
\lim_{n\to\infty} \frac1n\, \rel{D_H^\e}{\FF_n}{\rho^{\otimes n}} = D(\FF \| \rho)\qquad \forall\ \rho\in \D(\HH),\quad \forall\ \e\in (0,1)\, ;
\label{quantum_Sanov_D_H}
\ee
in particular, the Sanov exponent is given by the single-letter expression
\bb
\sanov(\rho\|\FF) = D(\FF \| \rho)\, .
\label{quantum_Sanov}
\ee
\end{thm}

Before we present the proof of Theorem~\ref{quantum_Sanov_thm}, it is instructive to see how it can be applied to calculate the Sanov exponent of entanglement testing. In what follows, $\SEP_{A:B}$ will denote the set of separable states~\cite{Werner} on the bipartite quantum system $AB$.

\begin{cor}
\label{cor:sep-sanov}
Let $AB$ be a finite-dimensional bipartite quantum system. Then, it holds for all states $\rho_{AB}$ that
\bb
\label{eq:ent_sanov}
\lim_{n\to\infty} \frac1n\, \rel{D_H^\e}{\SEP_{A^n:B^n}}{\rho_{AB}^{\otimes n}} = \rel{D}{\SEP_{A:B}}{\rho_{AB}} = \!\inf_{\sigma_{AB} \in \SEP_{A:B}} D(\sigma_{AB}\|\rho_{AB}) \qquad 
\forall\ \e\in (0,1)\, .
\ee
In particular, the Sanov exponent associated with entanglement testing, as well as the error exponent of entanglement distillation, are given by the single-letter formula
\bb
E_{d,\rm err} (\rho_{AB}) = \sanov\big(\rho\,\big\|\, \SEP_{A:B}\big) = \rel{D}{\SEP_{A:B}}{\rho_{AB}} \ .
\ee
\label{cor:sanov}
\end{cor}

Here we note that our solution of the quantum Sanov's theorem in fact shows a strong converse property --- the exponent does not depend on $\ve \in (0,1)$ and uniformly equals $\rel{D}{\SEP_{A:B}}{\rho_{AB}}$ for all $\ve$. One immediate consequence of this is that our result for entanglement distillation is stronger than stated in the main text; namely, \eqref{eq:ent_sanov} implies that
\bb
E_{d,\rm err}^{(m)} (\rho_{AB}) = \rel{D}{\SEP_{A:B}}{\rho_{AB}} \qquad \forall m \in \mathds{N} \ ,
\ee
where $E_{d, \rm err}^{(m)}$ is the $m$-copy distillation error exponent defined in~\eqref{eq:m-copy-exponent}. This shows that the distillation exponent does not actually depend on the target number of copies of the maximally entangled state at all.

\begin{proof}[Proof of Corollary~\ref{cor:sanov}]
Setting $r_n = d^{-n/2}$, we see that for an arbitrary operator $X_n = X_n^\dag \in \LL\big(\HH_{AB}^{\otimes n}\big)$ with $\|X_n\|_\infty \leq r_n$ we have that
\bb
\|X_n\|_2 \leq \sqrt{\dim \big(\HH_{AB}^{\otimes n}\big)}\, \|X_n\|_\infty = d^{n/2}\, \|X_n\|_\infty \leq 1\, ,
\ee
implying, via~\cite[Theorem~1]{GurvitsBarnum}, that $\id \pm X_n\in \cone\big(\SEP_{A^n:B^n}\big)$. The compatibility between the set of measurements defined by~\eqref{quantum_Sanov_thm_compatible_set} and the set of separable states thus follows from the well-known compatibility between the sets of separable measurements and that of separable states~\cite{Piani2009, brandao_adversarial}. Axiom~$6^\prime$ is therefore satisfied and we can thus immediately apply Theorem~\ref{quantum_Sanov_thm} to the case of entanglement testing.
\end{proof}

\begin{rem}
The compatibility assumption of Theorem~\ref{quantum_Sanov_thm} is quite versatile, and it can be shown to be obeyed by 
many resource theories of interest in quantum information and quantum computation. In all 
such cases, Theorem~\ref{quantum_Sanov_thm} provides us with the single-letter expression~\eqref{quantum_Sanov} for the generalised Sanov exponent. Beyond the set of separable states, an exemplary case of this extension that is also relevant in the study of quantum entanglement is provided by the set of positive partial transpose (PPT) states $\PPT_{\!A:B}$ --- see remark in Supplementary Note~\ref{app:lemma1} and discussion of compatibility in~\cite{Piani2009, brandao_adversarial}.

However, while these axiomatic assumptions can be verified for other sets of states relevant to quantum resource theories even beyond entanglement, it may not always be
possible to establish an equivalence between the error exponent of resource testing and the error exponent of distillation as in Lemma~\ref{lem:distillation_sanov_equivalence}. The study of other quantum resources for which this may be possible is an interesting open question.
\end{rem}

Having illustrated its main consequences, we present now the full proof our main result, the generalised quantum Sanov's theorem.

\begin{proof}[Proof of Theorem~\ref{quantum_Sanov_thm}]
Given what we have already proved, it is not difficult to see that~\eqref{quantum_Sanov_D_H} and~\eqref{quantum_Sanov} are actually equivalent. In fact, due to~\eqref{weak_strong_converse_duality} and Proposition~\ref{limit_e_to_0_prop} we know that
\bb
\lim_{\e\to 1^-} \lim_{n\to\infty} \frac1n\, \rel{D_H^\e}{\FF_n}{\rho^{\otimes n}} = D(\FF \| \rho)\qquad \forall\ \rho\in \D(\HH)\, .
\ee
Since $D_H^\e$ is monotonically non-decreasing in $\e$, to establish~\eqref{quantum_Sanov_D_H} it thus suffices to prove that also 
\bb
\sanov(\rho\|\FF) = \lim_{\e\to 0^+} \lim_{n\to\infty} \frac1n\, \rel{D_H^\e}{\FF_n}{\rho^{\otimes n}} \geqt{?} D(\FF \| \rho)\qquad \forall\ \rho\in \D(\HH)\, .
\label{quantum_Sanov_nontrivial_inequality}
\ee
From now on, we will therefore focus on proving the above inequality.

Fix a positive integer $k\in \N^+$. Let $\all$ be the set of all quantum measurements with finitely many outcomes on any given quantum system. Since this set is closed under finite labelled mixtures, according to Lemma~\ref{minimax_measured_relent_lemma} we have that
\bb
\rel{D^{\all}}{\FF_k}{\rho^{\otimes k}} = \sup_{\MM\in \all} \inf_{\sigma_k \in \FF_k} \rel{D}{\MM(\sigma_k)}{\MM\big(\rho^{\otimes k}\big)}\, .
\ee
For some $\eta > 0$, let $\MM_*$ be an almost optimal measurement for the supremum on the right-hand side of the above equation, i.e.\ let it satisfy
\bb
\rel{D}{\MM_*(\FF_k)}{\MM_*\big(\rho^{\otimes k}\big)} \geq \sup_{\MM\in \all} \inf_{\sigma_k \in \FF_k} \rel{D}{\MM(\sigma_k)}{\MM\big(\rho^{\otimes k}\big)} - \eta = \rel{D^{\all}}{\FF_k}{\rho^{\otimes k}} - \eta\, .
\label{choice_M_*}
\ee
We denote by $(E_x)_{x\in \XX}$ the POVM that represents $\MM_*$. Without loss of generality, we can assume that $E_x\neq 0$ for all $x\in \XX$.

We now have all the elements to explain our strategy to carry out hypothesis testing at the quantum level. We proceed as follows:
\begin{enumerate}[(1)]
\item We partition the $n$ systems into $n' \coloneqq \floor{n/k}$ groups of $k$ systems each, discarding the rest.
\item We measure each group with the measurement $\MM_*$, obtaining a sequence $x^{n'}$ of $n'$ classical symbols $x_i\in \XX$.
\item We perform classical hypothesis testing on this sequence with the goal of minimising the type I error.
\end{enumerate}
\emph{Provided that we can apply the generalised classical Sanov's theorem} (Theorem~\ref{Sanov_thm}), the above strategy gives the following estimate on the quantum Sanov exponent:
\bb
\sanov(\rho\|\FF)\ &\geqt{($\star$)}\ \frac{1}{k}\, \rel{D}{\MM_*(\FF_k)}{\MM_*\big(\rho^{\otimes k}\big)} \\
&\geq\ \frac{1}{k}\, \rel{D^{\all}}{\FF_k}{\rho^{\otimes k}} - \frac{\eta}{k} \\
&=\ \frac{1}{k}\, \inf_{\sigma_k\in \FF_k} \rel{D^{\all}}{\sigma_k}{\rho^{\otimes k}} - \frac{\eta}{k} \\
&\geqt{(i)}\ \frac{1}{k}\, \inf_{\sigma_k\in \FF_k} \rel{D}{\sigma_k}{\rho^{\otimes k}} - \frac1k \log_2 \frac{(k+d-1)^{d-1}}{(d-1)!} - \frac{\eta}{k} \\
&\eqt{(ii)}\ D(\FF\|\rho) - \frac1k \log_2 \frac{(k+d-1)^{d-1}}{(d-1)!} - \frac{\eta}{k}\, .
\label{quantum_Sanov_proof_key_inequality}
\ee
Here, ($\star$) comes from the application of Theorem~\ref{Sanov_thm}, whose legitimacy we will verify shortly. In~(i) we used instead the entropic pinching inequality~\eqref{entropic_pinching_k_copies}, denoting with $d\coloneqq \dim \HH$ the dimension of $\HH$, while (ii)~follows by the additivity of the reverse relative entropy of resource. Taking the limit $k\to\infty$ of~\eqref{quantum_Sanov_proof_key_inequality} would thus yield precisely~\eqref{quantum_Sanov_nontrivial_inequality} and therefore complete the proof.

We will now argue that the inequality ($\star$) in~\eqref{quantum_Sanov_proof_key_inequality} holds, because we can legitimately apply Theorem~\ref{Sanov_thm}. To this end, we construct the sequence of sets of probability distributions $\big(\widetilde{\FF}_n\big)_n$ given by
\bb
\widetilde{\FF}_n \coloneqq \MM_*^{\otimes n} ( \FF_{nk}	) \subseteq \PP(\XX^n)\, .
\ee
(We omit the dependence on $k$ for simplicity.) If $(\FF_n)_n$ satisfies the Brand\~{a}o--Plenio axioms, as we have assumed, it is easy to verify that so does $\big(\widetilde{\FF}_n\big)_n$. Let us verify this claim:

\begin{enumerate}
\item The convexity and closedness of $\widetilde{\FF}_n$ follows from the convexity and closedness of $\FF_{nk}$, due to the finite dimensionality and to the linearity of the copy-by-copy measurement map $\MM_*^{\otimes n}$.
\item Due to Axiom~4 holding for $(\FF_n)_n$, the set $\widetilde{\FF}_1 = \MM_*(\FF_{k})$ contains the probability distribution $q_0 \coloneqq \MM_* \big(\sigma_0^{\otimes k}\big)$, where $\sigma_0 \in \FF_1$ is the full-rank state whose existence is guaranteed by Axiom~2. Note that $\sigma_0^{\otimes k}$ is full rank; thus, since $E_x\neq 0$ for all $x\in \XX$ by assumption, it holds that $q_0(x) = \Tr\big[ \sigma_0^{\otimes k} E_x \big] > 0$ for all $x\in \XX$, i.e.\ $q_0$ has full support.

\item For a probability distribution $q = \MM_*^{\otimes (n+1)}(\sigma)\in \widetilde{\FF}_{n+1}$, where $\sigma \in \FF_{(n+1)k}$, by repeated applications of Axiom~3 we see that tracing away the last system of $q$ results in a probability distribution 
\bb
\tr_{n+1} q = \tr_{n+1} \MM_*^{\otimes (n+1)}(\sigma) = \MM_*^{\otimes n} \big( \Tr_{nk+1,\ldots, (n+1)k} \sigma\big) \in \MM_*^{\otimes n} (\FF_{nk}) = \widetilde{\FF}_n\, . 
\ee

\item If $q = \MM_*^{\otimes n}(\sigma)\in \widetilde{\FF}_n$ and $r = \MM_*^{\otimes m}(\omega)\in \widetilde{\FF}_m$, with $\sigma\in \FF_{nk}$ and $\omega\in \FF_{mk}$, then clearly
\bb
q \otimes r = \MM_*^{\otimes (n+m)}(\sigma \otimes \omega) \in \widetilde{\FF}_{n+m}
\ee
simply because $\sigma \otimes \omega \in \FF_{nk+mk} = \FF_{(n+m)k}$. This establishes that $\big(\widetilde{\FF}_n\big)_n$ is indeed closed under tensor products.

\item Given $q = \MM_*^{\otimes n}(\sigma)\in \widetilde{\FF}_n$ and some permutation $\pi\in S_n$, we will call 
\bb
q_\pi(x^n) \coloneqq q\big(\pi^{-1}(x^n)\big) = q\big(x_{\pi(1)}\ldots x_{\pi(n)}\big)
\ee
the permuted probability distribution. Then, denoting with $\widetilde{U}_\pi$ the unitary that implements $\pi$ on each set of indices $\{i,i+k,\ldots, i+(n-1)k\}$, where $i=1,\ldots, k$ (such a unitary is, of course, isomorphic to $U_\pi^{\otimes k}$, with $U_\pi$ being the unitary that represents $\pi$ on $\C^{\otimes n}$), we have that
\bb
q_\pi(x^n) &= q\big(x_{\pi(1)}\ldots x_{\pi(n)}\big) \\
&= \Tr \big[ \big(E_{\pi(x_1)} \otimes \ldots \otimes E_{\pi(x_n)}\big)\, \sigma \big] \\
&= \Tr \big[ \widetilde{U}_\pi^\dag \big(E_{x_1} \otimes \ldots \otimes E_{x_n}\big) \widetilde{U}_\pi\, \sigma \big] \\
&= \Tr \big[ \big(E_{x_1} \otimes \ldots \otimes E_{x_n}\big)\, \widetilde{U}_\pi \sigma \widetilde{U}_\pi^\dag \big]
\ee
This shows that
\bb
q_\pi = \MM_*^{\otimes n}\big(\widetilde{U}_\pi \sigma \widetilde{U}_\pi^\dag\big) \in \MM_*^{\otimes n} (\FF_{nk}) = \widetilde{\FF}_n\, ,
\ee
where we observed that $\widetilde{U}_\pi \sigma \widetilde{U}_\pi^\dag \in \FF_{nk}$ due to Axiom~5.
\end{enumerate}

This proves that $\big(\widetilde{\FF}_n\big)_n$ does indeed satisfy the Brand\~ao--Plenio axioms. We now claim that the same family of sets satisfies also Axiom~6, because of our assumptions. To prove this claim, we will apply Piani's result in Lemma~\ref{Piani_lemma_to_verify_axiom_6}. To this end, we need to construct a sequence of sets of classical channels $\big(\mathds{L}_n\big)_n$ that obeys the compatibility condition in Definition~\ref{quantum_compatibility_def}, re-written in~\eqref{classical_compatibility} for classical systems, and such that $\mathds{L}_1$ is informationally complete, i.e.\ it satisfies~\eqref{information_completeness_classical}.

For each $n$ and each sequence $\widebar{x}^{\, n}\in \XX^n$, consider the classical channel $\Lambda_{\widebar{x}^{\,n}}: \R^{\XX^n} \to \R^2$, whose output $y\in \{0,1\}$ is binary, defined by the conditional probability distribution
\bb
\Lambda_{\widebar{x}^{\,n}}(y|x^n) \coloneqq \frac12 \left( 1 + (-1)^y r_{nk}\, \delta_{x^n,\, \widebar{x}^{\,n}} \right) ,
\label{Lambda_x_bar}
\ee
where $r_{nk}\in (0,1]$ is the number appearing in~\eqref{quantum_Sanov_thm_compatible_set}, and $\delta_{x^n,\, \widebar{x}^{\,n}} \coloneqq \prod_{i=1}^n \delta_{x_i,\, \widebar{x}_i}$, with $\delta$ representing the Kronecker delta. We now construct the sequence $\big(\mathds{L}_n\big)_n$ of sets of classical channels
\bb
\mathds{L}_n \coloneqq \left\{ \Lambda_{\widebar{x}^{\,n}}:\ \widebar{x}^{\,n}\in \XX^n \right\} .
\ee
Note that $\mathds{L}_1$ is informationally complete according to~\eqref{information_completeness_classical}, because
\bb
\Span \left\{ (\Lambda_{\widebar{x}}(y|x))_x:\ \widebar{x}\in \XX,\ y\in \{0,1\} \right\} = \R^\XX .
\label{quantum_Sanov_information_completeness}
\ee
To verify the above identity, it suffices to consider a vector $v\in \R^\XX$ that is orthogonal to the set on the left-hand side, i.e.\ it satisfies
\bb
0 = 2\sum_x v_x\, \Lambda_{\widebar{x}}(y|x) = \sum_x v_x + (-1)^y r_k v_{\widebar{x}} \qquad \forall\ \widebar{x}\in \XX,\quad \forall\ y\in \{0,1\}\, .
\ee
Summing over $y$ we obtain that $\sum_x v_x =0$, which plugged into the above identity shows that $v=0$. This proves~\eqref{quantum_Sanov_information_completeness}, establishing that $\mathds{L}_1$ is indeed informationally complete.

We now set out to verify the compatibility condition~\eqref{classical_compatibility}. Consider the free probability distribution
\bb
q_{n+m} = \MM_*^{\otimes (n+m)}\big(\sigma_{(n+m)k}\big) \in \widetilde{\FF}_{n+m}\, .
\ee
For some $\widebar{x}^{\,n}\in \XX^n$ and $\tilde{x}^m \in \XX^m$, we construct
\bb
\tilde{q}_m(\tilde{x}^m)\ \coloneqq&\ \ \sum_{x^n} \Lambda_{\widebar{x}^{\, n}}(y|x^n)\, q_{n+m}(x^n\tilde{x}^m) \\
=&\ \ \sum_{x^n} \Lambda_{\widebar{x}^{\, n}}(y|x^n)\, \Tr \left[ \big(E_{x_1}\otimes \ldots \otimes E_{x_n} \otimes E_{\tilde{x}_1} \otimes \ldots \otimes E_{\tilde{x}_m}\big)\, \sigma_{(n+m)k} \right] \\
=&\ \ \Tr\left[ \big(E_{\tilde{x}_1} \otimes \ldots \otimes E_{\tilde{x}_m}\big) \Tr_{1,\ldots, nk} \left[ \Big(\Big( \sumno_{x^n} \Lambda_{\widebar{x}^{\, n}}(y|x^n)\, E_{x_1}\otimes \ldots \otimes E_{x_n}\Big) \otimes \id^{\otimes mk}\Big)\, \sigma_{(n+m)k} \right] \right] \\
\eqt{(iii)}&\ \ \frac12 \Tr\left[ \big(E_{\tilde{x}_1} \otimes \ldots \otimes E_{\tilde{x}_m}\big) \Tr_{1,\ldots, nk} \left[ \Big(\Big( \id^{\otimes nk} + (-1)^y r_{nk}\, E_{\widebar{x}_1}\otimes \ldots \otimes E_{\widebar{x}_n} \Big) \otimes \id^{\otimes mk}\Big)\, \sigma_{(n+m)k} \right] \right] .
\label{quantum_Sanov_proof_q_tilde}
\ee
In~(iii) we used the definition~\eqref{Lambda_x_bar} and observed that 
\bb
\sum_{x^n} E_{x_1}\otimes \ldots \otimes E_{x_n} = \left(\sumno_x E_x\right)^{\otimes n} = \id^{\otimes nk}\, .
\ee
due to the normalisation condition for the POVM $(E_x)_x$. The calculation in~\eqref{quantum_Sanov_proof_q_tilde} proves that
\bb
\tilde{q}_m &= \MM_*^{\otimes m} \left( \Tr_{1,\ldots, nk} \left[ \bigg(\bigg( \frac{\id^{\otimes nk} + (-1)^y r_{nk}\, E_{\widebar{x}_1}\otimes \ldots \otimes E_{\widebar{x}_n}}{2} \bigg) \otimes \id^{\otimes mk}\bigg)\, \sigma_{(n+m)k} \right] \right) \\
&\hspace{-2.7pt}\overset{\text{(iv)}}{\in} \MM_*^{\otimes m} \left( \cone(\FF_{mk}) \right) \\
&= \cone\big(\MM_*^{\otimes m}(\FF_{mk})\big) \\
&= \cone\big(\widetilde{\FF}_m\big)\, ,
\ee
which proves that the sets of probability distributions $\big(\widetilde{\FF}_n\big)_n$ and the sets of classical channels $\big(\mathds{L}_n\big)_n$ satisfy the compatibility conditions in~\eqref{classical_compatibility}. In the above calculation, (iv)~follows from our assumption that the set of measurements in~\eqref{quantum_Sanov_thm_compatible_set} is compatible (in the quantum sense of Definition~\ref{quantum_compatibility_def}) with the sets of free states $(\FF_n)_n$, because
\bb
\left\| (-1)^y r_{nk}\, E_{\widebar{x}_1}\otimes \ldots \otimes E_{\widebar{x}_n} \right\|_\infty = r_{nk}\, \big\|E_{\widebar{x}_1}\big\|_\infty \ldots \big\|E_{\widebar{x}_n}\big\|_\infty \leq r_{nk}\, .
\ee

This concludes the proof that $\big(\widetilde{\FF}_n\big)_n$ satisfies not only the Brand\~{a}o--Plenio axioms, but also Axiom~6 in Section~\ref{sec_BP_axioms}. This means that the application of the generalised classical Sanov's theorem (Theorem~\ref{Sanov_thm}) marked as ($\star$) in~\eqref{quantum_Sanov_proof_key_inequality} is justified. According to the above discussion, taking the limit $k\to\infty$ in~\eqref{quantum_Sanov_proof_key_inequality} completes the proof.
\end{proof}

Similarly to what we did in Section~\ref{subsec:optimal_test_classical} for the classical case, also in the quantum case we can ask ourselves how to design a sequence of tests that achieves the optimal performance advertised in Theorem~\ref{quantum_Sanov_thm}. Namely, for any given $\rho$, the tests should discriminate $\rho^{\otimes n}$ from an arbitrary free state $\sigma_n\in \FF_n$, while exhibiting:
\begin{enumerate}[(a)]
\item a vanishing type II error probability; and
\item a type I error exponent arbitrarily close to the reverse relative entropy $D(\FF\|\rho)$.
\end{enumerate}

To construct such a test, one can proceed as follows. Denoting by $\e$ an overall tolerance on the final type I error exponent, we start by choosing some $\eta>0$ and some integer $k\in \N^+$ such that the last two terms on the rightmost side of~\eqref{quantum_Sanov_proof_key_inequality} are smaller than $\e/2$. This is certainly possible, as the sum of the two vanishes in the limit where $k\to\infty$ (here $d$, the underlying Hilbert space dimension, is fixed). We now need to find a measurement $\MM_*$ that satisfies~\eqref{choice_M_*}; since $\eta$ and $k$ have been fixed, and thus all involved states have bounded dimension, this can be accomplished numerically. (The complexity of this step, naturally, will depend on the precise nature of the sets of free states $\FF_k$, e.g.\ on whether they admit an efficient description.)

We can now proceed as described in the above proof: we divide the $n$ quantum systems into $n' = \floor{n/k}$ groups of $k$ copies each, discarding the remainder, and measure each group with $\MM_*$, obtaining a string of classical symbols $x^{n'}$. To this string we can apply the classically optimal test designed in Section~\ref{subsec:optimal_test_classical}, which consists in ascertaining whether $\frac12 \left\| t_{x^{n'}} - \MM_*(\FF_k) \right\|_1\leq \zeta$, where $\zeta>0$ is a small parameter, or not; in the first case, we guess that the underlying state was free, while in the second we guess that it was $\rho^{\otimes n}$. As explained in Section~\ref{subsec:optimal_test_classical}, the type II error probability associated with this test vanishes asymptotically as $n\to\infty$ (and thus $n'\to \infty$). The overall type I error exponent is at least
\bb
\frac1k\,\rel{D^\zeta}{\MM_*(\FF_k)}{\MM_*\big(\rho^{\otimes k}\big)} \geq \frac1k\,\rel{D}{\MM_*(\FF_k)}{\MM_k\big(\rho^{\otimes k}\big)} - \frac{\e}{2} \geq D(\FF\|\rho) - \e\, ,
\ee
where the first inequality holds due to Lemma~\ref{e_relent_convergence_lemma}, provided that we choose $\zeta>0$ to be sufficiently small, while the second holds by construction, due to the choice of $k$.


\section{Further considerations}
\label{app:further}

\subsection{Additivity violation for the reverse R\'enyi relative entropy of entanglement}
\label{app:additivity-violation}

We now set out to show that the R\'enyi-$\alpha$ reverse relative entropy of entanglement is never weakly additive for $\alpha < 1$. To this end, it suffices to take as a counterexample the antisymmetric Werner state $\rho_1$~\cite{Werner, Christandl2012}, humorously referred to as the `universal counterexample in entanglement theory'~\cite{Aaronson2008}. Consider the bipartite states on $\C^d\otimes \C^d$ given by
\bb
\rho_a \coloneqq \frac{\id + (-1)^a F}{d(d + (-1)^a)}\, ,\qquad a=0,1\, ,\qquad F \coloneqq \sum_{i,j=1}^d \ketbra{ij}{ji}\, ;
\ee
incidentally, these are the extremal points of the one-parameter family of Werner states $\rho(\delta) \coloneqq (1-\delta) \rho_0 + \delta \rho_1$, with $\delta\in [0,1]$~\cite{Werner}. In what follows, we will omit the dependence on $d$ for simplicity. 

\begin{lemma} \label{additivity_violations_lemma}
Let $\alpha \in (0,1
)$, and let $\mathds{D}_\alpha$ be any quantum divergence that reduces to the R\'enyi-$\alpha$ divergence for classical (i.e.\ commuting) states and obeys the data processing inequality. Then the corresponding reverse relative entropy of entanglement, defined by $\mathds{D}_\alpha(\SEP\|\rho) \coloneqq \min_{\sigma \in \SEP} \mathds{D}_\alpha (\sigma\|\rho)$, fails to be (weakly) additive on many copies of the antisymmetric state $\rho_1$, in formula
\bb
\mathds{D}_\alpha\big(\SEP\, \big\|\, \rho_1^{\otimes 2}\big) < 2\, \mathds{D}_\alpha(\SEP\|\rho_1)\, .
\ee
\end{lemma}

If $\mathds{D}_\alpha$ is also assumed to be additive, then the claim of the above lemma can be seen as a consequence of the results in~\cite[Section~5]{Rubboli2022}. To see why, recall that there it is argued that any function of the form $\rho\mapsto \min_{\sigma\in \SEP} \mathds{D}(\rho\|\sigma)$, where $\mathds{D}$ is a function on pairs of quantum states that obeys data processing, additivity, and normalisation --- meaning that $\mathds{D}\big(\ketbra{0} \big\|\, \tfrac12 (\ketbra{0}+\ketbra{1})\big) = 1$ --- fails to be additive on two copies of the antisymmetric state. Taking $\mathds{D}(\rho\|\sigma) = \tfrac{1-\alpha}{\alpha}\, \mathds{D}_\alpha(\sigma\|\rho)$, where $\mathds{D}_\alpha$ is a divergence as in the statement of Lemma~\ref{additivity_violations_lemma}, yields $\min_{\sigma\in \SEP} \mathds{D}(\rho\|\sigma) = \tfrac{1-\alpha}{\alpha}\, \mathds{D}_\alpha(\SEP\|\rho)$; assuming additivity, one then deduces the claim of Lemma~\ref{additivity_violations_lemma}. (This observation was communicated to us by R.\ Rubboli after the first version of our manuscript appeared~\cite{rubboli}.) Below, we present our independent proof of Lemma~\ref{additivity_violations_lemma}, as it is short, rather direct, and it by-passes the additivity requirement.

\begin{proof}[Proof of Lemma~\ref{additivity_violations_lemma}]
It is well known, via symmetry arguments (see, e.g., the discussion around~\cite[Eq.~(4)]{Audenaert2001}), that the closest separable state to $\rho(\delta)^{\otimes n}$, for any $\delta\in [0,1]$, is of the form 
\bb
\sigma = \sum_{a_1,\ldots, a_n\in \{0,1\}^n} P_{a_1,\ldots, a_n}\, \rho_{a_1}\otimes \ldots \rho_{a_n}\, ,
\label{twirled_sigma}
\ee
where $P$ is an arbitrary probability distribution on $\{0,1\}^n$. This is usually argued for the standard relative entropy of entanglement, with the optimisation with respect to separable states on the second argument, but it holds equally well for $\mathds{D}_\alpha(\SEP\|\rho)$ --- in fact, the proof only uses data processing and the fact that the Werner twirling map is separability preserving.

Now, for a state $\sigma$ of the form in~\eqref{twirled_sigma} we have that
\bb
\mathds{D}_\alpha\big(\sigma\, \big\|\, \rho_1^{\otimes n}\big) \eqt{(i)} D_\alpha(P \,\|\, e_{1^n}) \eqt{(ii)} - \frac{\alpha}{1-\alpha} \log_2 P_{1^n} \eqt{(iii)} \frac{\alpha}{1-\alpha}\, D\big(\rho_1^{\otimes n}\, \big\|\, \sigma\big) \, .
\ee
Here, in~(i) we made use of the fact that $\mathds{D}_\alpha$ must reduce to its classical version when the states commute --- and $\sigma$ and $\rho_1^{\otimes n}$ do commute --- in~(ii) we introduced the probability distribution $e_{1^n}$ defined by $e_{1^n}(x^n) = \prod_{i=1}^n \delta_{x_i,1}$, where $\delta_{a,b}$ is the Kronecker delta, and finally in~(iii) we considered the \emph{standard} relative entropy, with $\sigma$ in the second argument. Optimising over separable states $\sigma$ of the form in~\eqref{twirled_sigma} we obtain that
\bb
\mathds{D}_\alpha\big(\SEP\, \big\|\, \rho_1^{\otimes n}\big) = \frac{\alpha}{1-\alpha}\, D\big(\rho_1^{\otimes n}\, \big\|\, \SEP\big)\, .
\ee
Since it is well known that $D\big(\rho_1^{\otimes 2}\, \big\|\, \SEP\big) < 2\, D(\rho_1\| \SEP)$~\cite{Werner-symmetry}, the proof is complete.
\end{proof}


\subsection{A (classical) counter-example}
\label{app:counterexample}

Here, we construct a set of classical probability distribution that obeys the Brand\~ao--Plenio Axioms 1 -- 5 but whose Sanov exponent is different from the reverse relative entropy. In other words, this shows that the Brand\~ao--Plenio axioms do not suffice by themselves to guarantee that the generalized Sanov exponent is given by the reverse relative entropy. At the root of this discrepancy is the fact that in this model the exponent, i.e.\ the regularised relative entropy, can vanish even for resourceful states.

Consider a device that outputs a symbol $x\in \{0,1\}$ every time it is used. What you would really like it to do is to output the symbol 1 many times. Unfortunately, all that it does is that it sets itself in a random state 0 or 1, outputs the corresponding symbol some unknown number of times, and then it resets again in a random 0/1 state (as struck by a lightning, or a cosmic ray), and it continues in that way until it has produced a total of $n$ symbols. You then get all the symbols it produced, but shuffled in some random order. The corresponding set of probability distributions is given by
\bb
\FF_n \coloneqq \co\left\{ r_{J_1}\otimes \ldots \otimes r_{J_k}:\ \text{$J_1,\ldots, J_k$ is a partition of $[n]$} \right\} ,
\ee
where
\bb
r_J\big(x^J\big) \coloneqq \left\{ \begin{array}{ll} 1/2 & \text{ if $x^J = 0^J$ or $x^J = 1^J$,} \\[1ex] 0 & \text{ otherwise.} \end{array}\right.
\ee

\begin{lemma}
The family of sets $\FF_n$ satisfies the Brand\~{a}o--Plenio Axioms 1 -- 5.
\end{lemma}

\begin{proof}
Convex, closedness, and permutation invariance are clear from the definition. Closure under tensorisation follows from the fact that joining a partition of $[n]$ and one of $[m]$ yields a partition of $[n+m]$. The uniform probability distribution $r_1$ on $\{0,1\}$ is in $\FF_1$ and has full support. 
It only remains to verify closure under partial trace. This is also easy, and it follows from the observation that $\tr_k r_k = r_{k-1}$ for all $k\geq 2$, where $\tr_k$ denotes the operation of discarding the last symbol.
\end{proof}

We are now ready to present our counter-example. Consider the probability distribution $e_1$ with the property that $e_1(0) = 0$ and $e_1(1) = 1$. Clearly, $e_1\notin \FF_1$ is not a free distribution --- after all, it corresponds to what we would like the device to do, not to what it does in reality. However,
\bb
D\big(e_1^{\otimes n}\,\big\|\, \FF_n\big) \leq D_{\max}\big(e_1^{\otimes n}\,\big\|\, \FF_n\big) \leq  D_{\max}(e_1^{\otimes n}\| r_n) = 1\, ,
\ee
implying that $D^\infty(e_1\| \FF) = 0$.

What about the reverse quantities? Rather trivially, $D(\FF\|e_1) = +\infty$ because $e_1$ is a deterministic probability distribution which is not in $\FF_1$. The same happens (as it should, because the reverse relative entropy is additive) for $n$ copies. What is even more interesting is that one can evaluate the hypothesis testing relative entropy $\min_{q_n\in \FF_n} D_H^\e(q_n\|e_1^{\otimes n})$ explicitly. Indeed, the best test $a: \{0,1\}^n\to [0,1]$ satisfies $a(x^n) = 1$ for $x^n\neq 1^n$; its only un-specified entry, $a(1^n)$, can be determined by requiring that 
\bb
1-\e \leq \sum_{x^n} a(x^n) q_n(x^n) = a(1^n) q_n(1^n) + 1 - q_n(1^n)\, ,
\ee
which immediately gives us $a(1^n) = \left( 1 - \frac{\e}{q_n(1^n)} \right)_+$. This in turn shows that
\bb
D_H^\e\big(q_n\,\big\|\, e_1^{\otimes n}\big) = - \log_2 \sum_{x^n} a(x^n)\, e_1^{\otimes n}(x^n) = - \log_2 a(1^n) = \left\{ \begin{array}{ll} -\log_2 \left( 1 - \frac{\e}{q_n(1^n)} \right) & \text{ if $\e < q_n(1^n)$,} \\[1.5ex] +\infty & \text{ otherwise,} \end{array} \right.
\ee
Since this is a non-increasing function of $q_n(1^n)$, it is minimised when this number is as large as possible, which happens when $q_n=r_n$, in which case it evaluates to $1/2$. Hence,
\bb
\min_{q_n\in \FF_n} D_H^\e\big(q_n\,\big\|\, e_1^{\otimes n}\big) = D_H^\e\big(r_n\,\big\|\, e_1^{\otimes n}\big) = \left\{ \begin{array}{ll} -\log_2 \left( 1 - 2\e \right) & \text{ if $\e < 1/2$,} \\[1.5ex] +\infty & \text{ $\e\geq 1/2$.} \end{array} \right.
\ee
In particular,
\bb
\liminf_{n\to\infty} \frac1n\, \min_{q_n\in \FF_n} D_H^\e\big(q_n\,\big\|\, e_1^{\otimes n}\big) = \left\{ \begin{array}{ll} 0 & \text{ if $\e < 1/2$,} \\[1.5ex] +\infty & \text{ $\e\geq 1/2$,} \end{array} \right.
\ee
entailing that
\bb
\mathrm{Sanov}(e_1 \| \FF) =\lim_{\e\to 0^+} \liminf_{n\to\infty} \frac1n\, \min_{q_n\in \FF_n} D_H^\e\big(q_n\,\big\|\, e_1^{\otimes n}\big) = 0 < \infty = D(\FF\|e_1)\, .
\ee

For completeness, we also consider the Stein setting. It is straightforward to verify that
\bb
D_H^\e\big(e_1^{\otimes n}\,\big\|\, q_n\big) = - \log_2(1-\e) - \log_2 q_n(1^n)\, ,
\ee
so that $\min_{q_n\in \FF_n} D_H^\e\big(e_1^{\otimes n}\,\big\|\, q_n\big) = - \log_2(1-\e) - 1$,
implying that
\bb
\liminf_{n\to\infty} \frac1n\, \min_{q_n\in \FF_n} D_H^\e\big(e_1^{\otimes n} \,\big\|\, q_n\big) = 0\qquad \forall\ \e\in [0,1)\, ,
\ee
and in turn that $\mathrm{Stein}(e_1\|\FF) = 0 = D^\infty(e_1\|\FF)$.

\end{document}